\documentclass{article}

\usepackage{etoolbox} %
\usepackage{tikz}
\usetikzlibrary{matrix}
\usetikzlibrary{external}
\usetikzlibrary{arrows.meta}

\usepackage{caption}
\usepackage{subcaption}

\providebool{ignoretikz}
\setbool{ignoretikz}{false}

\providecommand{\providesetbool}[2]{%
  \ifcsundef{if#1}{%
    \newbool{#1}%
    \setbool{#1}{#2}%
  }{}%
}

\providesetbool{externalizetikz}{false}

\ifbool{externalizetikz}{%
  \tikzsetexternalprefix{}
  \tikzexternalize[]
}{}

\newcommand{\redtt}[1]{{\color{red}\ttfamily #1}}
\newcommand{\inputtikzwithfilename}[1]{%
  \ifbool{ignoretikz}{%
    \redtt{Ignoring tikz: \detokenize{#1}}%
  }{%
    \IfFileExists{#1}{%
      \tikzsetnextfilename{externalized/#1}%
      \input{#1}%
    }{%
      \redtt{File not found: \detokenize{#1}}%
    }%
  }%
}

\newcommand{\inputtikzR}[1]{
  \input{tikzR/#1}%
}

\usepackage{geometry}
\usepackage{float} %
\usepackage{microtype} %

\usepackage{authblk}

\usepackage{quotes}

\usepackage{amsmath, amssymb, amsthm, amsfonts}
\usepackage{bbm} %

\usepackage{hyperref}

\hypersetup{
	colorlinks=true,
	linkcolor=blue,
	filecolor=magenta,
	urlcolor=cyan,
	citecolor=blue
}

\usepackage{etoolbox} %

\usepackage{xcolor}

\usepackage{algorithm}
\usepackage{algpseudocode}

\providebool{cleanbuild}

\providecommand{\comment}{}
\renewcommand{\comment}[2]{%
    \ifbool{cleanbuild}{}{%
        \ifmmode%
            \text{\textcolor{#1}{\sffamily\small[{#2}]}}%
        \else%
            \noindent%
            \textcolor{#1}{\sffamily\small[{#2\unskip}]}%
        \fi%
    }%
}%

\newcommand{\CAP}{\expandafter\MakeUppercase}
\newcommand{\HR}{Hüsler--Reiss}

\newcommand{\multGP}{multivariate generalized Pareto}

\newcommand{\pd}{positive definite}

\newcommand{\psd}{positive semidefinite}

\newcommand{\cnd}{conditionally negative definite}

\newcommand{\pppd}{parallelepiped}

\newcommand{\myleft}{\mathopen{}\mathclose\bgroup\left}
\newcommand{\myright}{\aftergroup\egroup\right}

\newcommand{\mlr}[1]{\myleft(#1\myright)}
\newcommand{\lr}[1]{(#1)}

\newcommand{\set}[1]{\{#1\}}

\newcommand{\setm}[2]{\{#1 \mid #2\}}

\newcommand{\norm}[2][]{\Vert#2\Vert_{#1}}

\newcommand{\abs}[1]{|#1|}

\newcommand{\brackets}[1]{\lbrack#1\rbrack}

\newcommand{\halfopen}[1]{\lbrack#1)}

\newcommand{\N}{\mathbb{N}}
\newcommand{\R}{\mathbb{R}}
\newcommand{\Rd}[1][d]{\R^{#1}}
\newcommand{\Rdd}[1][d]{\R^{#1 \times #1}}

\newcommand{\E}{\mathbb{E}} %
\providecommand{\P}{}
\renewcommand{\P}[1][]{\mathbb{P}_{#1}}
\newcommand{\Var}{\mathrm{Var}} %

\newcommand{\normal}{\mathcal{N}} %
\newcommand{\Exp}{\mathrm{Exp}} %

\newcommand{\Dcal}{\mathcal{D}}

\newcommand{\Lcal}{\mathcal{L}}
\newcommand{\Hcal}{\mathcal{H}}
\newcommand{\Ical}{\mathcal{I}}
\newcommand{\Jcal}{\mathcal{J}}

\newcommand{\inftyvec}{\boldsymbol{\infty}}
\newcommand{\zerovec}{\mathbf{0}}
\newcommand{\onevec}{\mathbf{1}}
\newcommand{\dinvvec}{d\inv\onevec}
\newcommand{\indicator}[1]{\mathbbm{1}\set{#1}}
\newcommand{\kev}[1]{\mathbf{e}_{#1}} %

\newcommand{\Id}{\mathrm{I}} %

\newcommand{\without}[1]{\setminus\set{#1}}

\newcommand{\deq}{\stackrel{d}{=}}

\newcommand{\slr}[1]{^{\lr{#1}}}

\newcommand{\T}{^\top\!} %
\newcommand{\TT}{^\top} %
\newcommand{\inv}{^{-1}}

\newcommand{\pinv}{^+}

\newcommand{\orth}{^{\perp}}

\renewcommand{\det}[1]{\abs{#1}}

\newcommand{\pdet}[1]{\abs{#1}_+}

\newcommand{\trace}{\mathrm{tr}}
\newcommand{\laspan}{\mathrm{span}}
\newcommand{\image}{\mathrm{Im}}
\newcommand{\kernel}{\mathrm{ker}}

\newcommand{\half}{\tfrac{1}{2}}
\DeclareMathOperator*{\argmin}{arg\,min}

\usepackage{etoolbox} %
\providebool{cleanbuild}

\usepackage{amsthm}

\usepackage{hyperref}
\usepackage[capitalize,nameinlink]{cleveref} %
\crefformat{equation}{#2(#1)#3}
\crefformat{figure}{#2Figure~#1#3}
\crefmultiformat{figure}{#2Figures~#1#3}{ and~#2#1#3}{, #2#1#3}{ and~#2#1#3}
\crefrangeformat{figure}{#2Figures~#1#3--#4#2#5}
\AddToHook{cmd/appendix/before}{\crefalias{section}{appendix}}

\usepackage{natbib}

\usepackage{aliascnt}
\usepackage{xparse}
\let\oldtheorem\newtheorem
\RenewDocumentCommand{\newtheorem}{m o m o}{%
    \IfValueTF{#2}{%
        \IfNoValueTF{#4}{%
            \newaliascnt{#1}{#2}%
            \oldtheorem{#1}[#1]{#3}%
            \aliascntresetthe{#1}%
        }{%
            \oldtheorem{#1}[#2]{#3}[#4]%
        }%
    }{%
        \IfNoValueTF{#4}{%
            \oldtheorem{#1}{#3}%
        }{%
            \oldtheorem{#1}{#3}[#4]%
        }%
    }%
}

\makeatletter \renewcommand\hyper@natlinkbreak[2]{#1} \makeatother

\numberwithin{equation}{section}

\theoremstyle{plain}
\newtheorem{theorem}{Theorem}[section]
\newtheorem{proposition}[theorem]{Proposition}
\newtheorem{lemma}[theorem]{Lemma}
\newtheorem{corollary}[theorem]{Corollary}
\newtheorem{definition}[theorem]{Definition}

\theoremstyle{definition}
\newtheorem{example}[theorem]{Example}

\theoremstyle{remark}
\newtheorem{remark}[theorem]{Remark}

\newcommand{\provenlabel}[1]{%
  \label{#1}%
  \ifbool{cleanbuild}{}{%
    \makeatletter%
    \ifcsname r@proof:#1\endcsname%
      \hyperlink{proof:#1}{(See Proof.)}%
    \fi%
    \makeatother%
  }%
}

\newcommand{\proofref}[1]{%
  \label{proof:#1}%
  Proof of \hypertarget{proof:#1}{\cref{#1}}%
}

\usepackage{enumitem}
\setlist[enumerate,1]{label={(\arabic*)}}

\begin{document}

\title{Directional variograms for multivariate extremes}

\author[1]{Manuel Hentschel}
\author[2]{Frank Röttger}
\author[3]{Johan Segers}
\author[1]{Sebastian Engelke}

\affil[1]{University of Geneva}
\affil[2]{University of Twente}
\affil[3]{KU Leuven and UCLouvain}

\date{03.07.2026}
\maketitle

\begin{abstract}
    \noindent%
    Multivariate generalized Pareto distributions arise as limits of threshold exceedances
and form a central model class for multivariate extremes.
Existing inference methods based on the extremal variogram condition on the value of a single component,
which can be statistically suboptimal.
We generalize this approach by conditioning the multivariate generalized Pareto random vector $Y$ to lie on arbitrary half-spaces.
Specifically, for a direction vector $v$,
we introduce the random vector $Y^v = (Y \mid v^\top Y > 0)$
and define the associated $v$-variogram $\Gamma_{ij}^v=\mathrm{Var}(Y_i^v-Y_j^v)$.
We establish the decomposition $Y^v \stackrel{d}{=} W^v+E\mathbf{1}$
into the so-called $v$-extremal function~$W^v$
and an independent exponential random variable~$E$,
and derive several results relating these random variables to each other.
For logistic, Dirichlet, and Hüsler--Reiss multivariate generalized Pareto models,
we derive closed-form expressions for~$\Gamma^v$.
In the Hüsler--Reiss case,
we further derive new density representations
and identify a distinguished resistance-curvature vector $v_0$
that uniquely centers the Gaussian law of $W^{v_0}$ while characterizing the least-mass half-space.
On the statistical side, we
introduce empirical $v$-variograms and show in a simulation study that the choice of $v$
induces a pronounced bias-variance trade-off that is
strongly related to the mass of the conditioning half-space.
Moreover, combining information across multiple directions $v$ can substantially reduce estimation variance
relative to methods based on a single vector.

\end{abstract}

\clearpage

\section{Introduction}
\label{sec:introduction}

Multivariate extreme value statistics investigates rare events in large systems with the aim of describing, estimating, and predicting their dependence structure and  frequency of occurrence.
Such methods are essential in applications where simultaneous or cascading extremes play a critical role, including environmental risk assessment, financial stability analysis, and engineering safety.
Within this framework, the approach of threshold exceedances considers a vector as extreme whenever at least one of its components exceeds a high threshold.

Multivariate generalized Pareto distributions appear as the only possible
limit of threshold exceedances and are therefore natural models for
extreme events \citep{rootzen2006, kiriliouk2018POT}.
The left-hand side of \cref{figure:intro} shows a sample from a multivariate Pareto distribution $Y = (Y_1,\dots, Y_d)$,
where every point exceeds the threshold in at least one component.
The non-standard support of the random vector
$Y$ is often problematic, and several papers consider
auxiliary vectors  $Y^{(m)} = (Y \mid Y_m > 0)$ on a product-form set for some $m=1,\dots ,d$
to simplify estimation \citep{engelke2015} or enable the definition
of extremal conditional independence \citep{segers2020, engelke2020}.
Based on these auxiliary vectors, \cite{engelkeVolgushev2022} introduce a
dependence measure
called the extremal variogram $\Gamma_{ij}^{(m)} =\Var(Y_i^{(m)} - Y_j^{(m)})$, $i,j=1,\dots , d$.
The empirical variogram $\hat \Gamma^{(m)}$ is the sample version of $\Gamma^{(m)}$, which however only considers the subset of observations $y\in\Rd$ of $Y$ that satisfy $y_m>0$.
In order to use all data, the most common estimator therefore averages over all components $m$ to yield $\hat \Gamma = \frac1m \sum_{m=1}^d \hat \Gamma^{(m)}$.
The conditioning along the axis is an arbitrary choice that might be suboptimal for estimation in terms of bias and/or variance.

In this work, for any vector $v$ in the $d$-variate probability simplex, we introduce the random vector $Y^{v} = (Y \mid v\T Y > 0)$ and the corresponding $v$-variogram matrix $\Gamma^{v}$ with entries
\begin{equation*}
    \Gamma_{ij}^{v}
    =
    \Var\lr{Y_i^v - Y_j^v}
    , \quad i,j=1,\dots , d,
\end{equation*}
where the original extremal variograms appear as special cases using the canonical
unit vectors for $v$.
For several of the most commonly used parametric models for multivariate
Pareto distributions, namely the logistic, Dirichlet \citep{CT1991}, and \HR{} \citep{hueslerReiss1989} models, we obtain closed-form expressions for $\Gamma^{v}$.
We further derive the stochastic representation
\begin{equation*}
    Y^v
      \deq
    W^v + E \onevec,
\end{equation*}
where $W^v\deq P_v Y^v$ is obtained by applying the oblique projection $P_v = \Id - \onevec v\T$ along $\onevec$ onto $v\orth$ to the random vector $Y^v$.
The vector $W^v$ is called the $v$-extremal function and is a generalization of the classical extremal function obtained when $v$ equals a unit vector \citep{dombry2013}.
The distributions and densities (if they exist) of the
random vectors $Y^{v}$ and $W^v$ for different vectors $v$ can be
related through exponential tilting.

We study the important case of the \HR{} distribution in more detail.
This model is parameterized by a variogram matrix $\Gamma$ and the $v$-extremal function follows a $(d-1)$-dimensional Gaussian
distribution on the hyperplane $v\T$.
The $v$-variogram is in this case independent of the choice of $v$ and satisfies $\Gamma^v = \Gamma$.
As a main theoretical contribution, we derive new density representations
for $Y$ and $Y^v$ in this model class and provide closed forms for the normalizing constants.
We further identify a special vector $v_0$, called the resistance curvature vector,
which corresponds to the half-space $\setm{y}{v_0\T y > 0}$ that contains
the lowest expected number of samples from $Y$.
At the same time,
it is the only vector such that the Gaussian distribution of $W^{v_0}$ is
centered, and leads to a particularly elegant density expression.
The vector $v_0$ has also a geometric interpretation when considering $\Gamma$ as a Euclidean distance matrix \citep{devriendt2022a}.

Estimation of the empirical version $\hat \Gamma^{v}$ of the $v$-variogram matrix is based on exceedances in the set $\setm{y}{v\T y > 0}$,
as shown by the filled points in the left panel of \cref{figure:intro} for the example of $v = (0.3,0.7)\T$ in $d=2$ dimensions.
The points in the right panel of \cref{figure:intro} represent squared biases and variances of estimates of the $v$-variogram matrices,
averaged over all entries, for different vectors $v$.
The cyan points correspond to the original
extremal variograms estimates $\hat \Gamma^{(m)}$, and can be seen to
have the largest biases but the smallest variances;
the latter can be explained by the fact that the corresponding half-spaces have the largest expected sample size.
On the other extreme, the estimator $\hat \Gamma^{v_0}$ of the resistance curvature vector, shown in orange, exhibits the smallest bias but the largest variance.
The squares show ensemble estimators based on different sets of vectors $v$.
These combined variograms can be seen to have the lowest mean squared error, indicated by the gray line.

We provide an extensive simulation study for inference based on our new $v$-variogram, with focus on models in the domain of attraction of a \HR{} distribution.
We observe that depending on the choice of the threshold, different $v$-variograms are optimal in terms of mean squared error.
Furthermore, we find that ensembles of $v$-variograms permit a significant improvement of performance, and that the choice of the optimal ensemble depends on the threshold.

\begin{figure}
    \centering
    \resizebox{0.98\textwidth}{!}{%
        \includegraphics{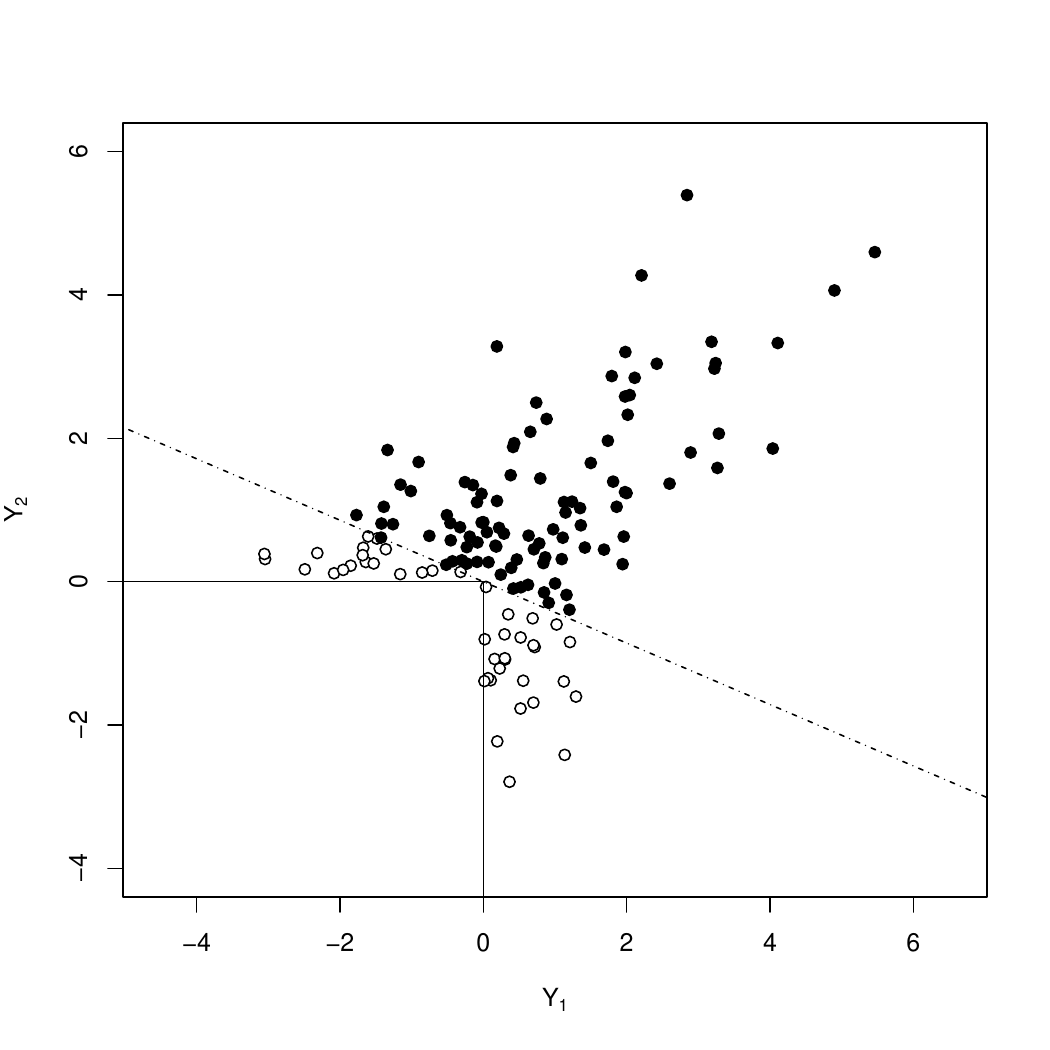}
        \includegraphics{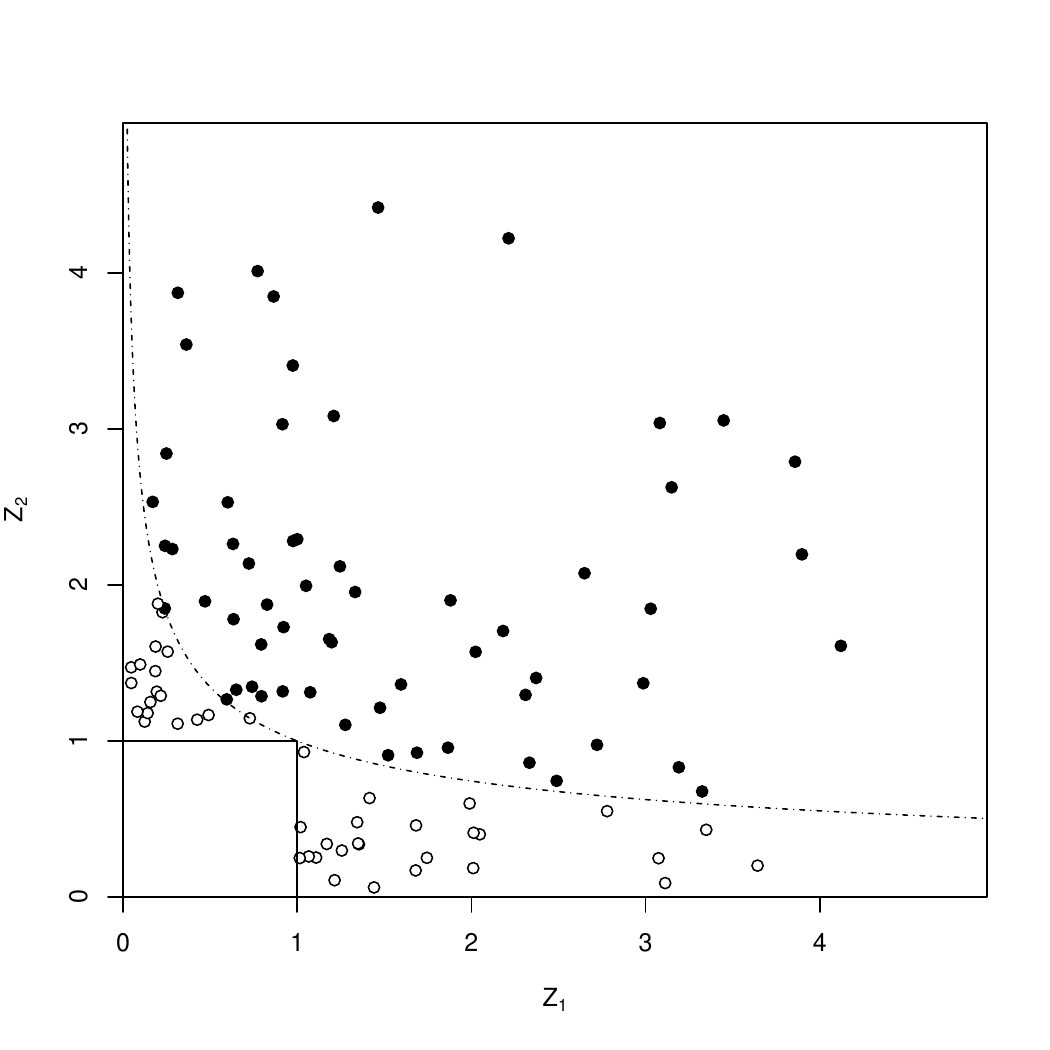}
        \includegraphics{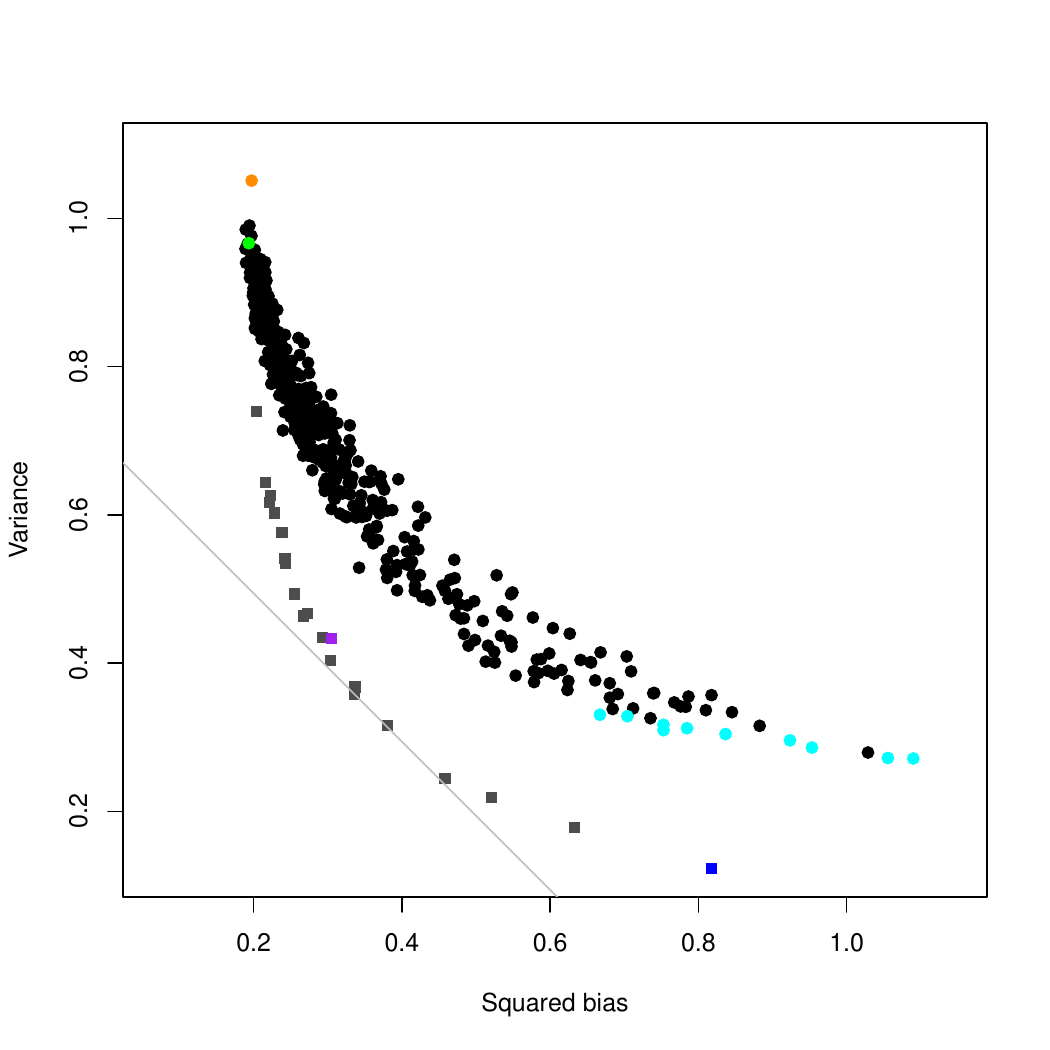}
    }
    \caption{%
        Left: sample from a \HR{} multivariate Pareto distribution in $d=2$ dimensions on exponential margins (all points),
        where filled points fall in the set $\setm{y}{v\T y > 0}$ for $v = (0.3,0.7)\T$.
        Center: the same sample on Pareto margins, obtained by the transformation $Z = \exp(Y)$.
        Right: performance of different $v$-variogram estimators for data in the domain of attraction of a \HR{} distribution.
        Ensemble estimators are shown as squares, and the MSE level curve in gray.
    }
    \label{figure:intro}
\end{figure}

\section{Preliminaries}
\label{sec:preliminaries}

\subsection{Projections and half-spaces}
\label{sec:linear_algebra}

For $v \in \Rd$, $v \neq \zerovec$
we define the hyperplane $v\orth = \setm{x \in \Rd}{v\T x = 0}$
and the half-space $\Hcal^v = \setm{x \in \Rd}{v\T x > 0}$.
For vectors satisfying $\onevec\T v \neq 0$,
we denote the oblique projection
along the $\onevec$-direction onto $v\orth$
as $P_v = \Id - \onevec v\T / (\onevec\T v)$.
Note that composed oblique projections
with the same kernel
simplify
to the last one, i.e., $P_v P_u = P_v$ for any $\onevec\T v, \onevec\T u > 0$.
Of particular interest will be
vectors from the probability simplex $\Delta_{d-1} = \setm{v \in \Rd}{\onevec\T v = 1, v \geq \zerovec}$
or more generally the affine hyperplane $\setm{v}{\onevec\T v = 1}$.
Note that for these vectors,
the expression for the oblique projection matrix simplifies to $P_v = \Id - \onevec v\T$.

\subsection{Multivariate generalized Pareto distribution}
\label{sec:MPD}

Let $X=(X_1,\ldots,X_d)\T$ be a random vector with standard exponential margins, that is, $\P\lr{X_i \leq x_i} = 1- \exp(-x_i)$ for $i = 1, \dots, d$ and $x_i \geq 0$.
We assume standard exponential margins in order to focus on the extremal dependence structure.
In practice we may need to employ empirical marginal transformations, leading to a rank-based procedure, which is done in extreme value theory at least since \citet{drees1998best}.
The choice of standard exponential margins over the more common standard Pareto margins simplifies the notation for the objectives of this paper.
For any $z\in \Lcal = \setm{x \in \Rd}{x \not\leq \zerovec}$, the limit of threshold exceedances
\begin{equation}
    \label{eq:mpd}
    \P\lr{Y \leq z}
    :=
    \lim_{u\to\infty} \P\lr{
        X - u\onevec \leq z
        \mid
        X \not\leq u \onevec
    }
    =
    \frac{
        \Lambda^c\lr{z \land \zerovec} - \Lambda^c\lr{z}
    }{
        \Lambda^c\lr{\zerovec}
    }
    ,
\end{equation}
if it exists, follows a multivariate generalized Pareto distribution
\citep{rootzen2006}.
We then say that the random vector $X$ is in the domain of attraction of the multivariate generalized Pareto random vector~$Y=(Y_1,\ldots,Y_d)\T$.
The distribution is characterized by an exponent measure $\Lambda$,
and we write $\Lambda^c(z) := \Lambda(\halfopen{-\infty,\infty}^d \setminus \brackets{-\inftyvec, z})$.
Because we started with normalized
margins of $X$, the measure $\Lambda$ is also normalized in the
sense that $\Lambda(y_i > 0) = 1$ for all $i=1,\dots, d$.
Moreover, it is homogeneous in the sense that $\Lambda(t\onevec + B) = \exp(-t) \Lambda(B)$.
If the exponent measure allows a density with respect to the Lebesgue measure, we call this density $\lambda(y)$,
and this density then satisfies $\lambda(y + t\onevec) = \exp(-t) \lambda(y)$ for all $y \in \Rd$ and all $t \in \R$.

A different but equivalent perspective on multivariate extremes
is through point processes.
We say that a random vector $U$ on $\Rd$
with $\E(e^{U_i}) = 1$ for $i=1,\dots,d$,
is a generator of $Y$,
if the corresponding exponent measure $\Lambda$ is
the intensity of the Poisson point process \citep[e.g.,][Proposition~5.8]{resnick2008extreme}
\begin{equation*}
    \Pi
    =
    \sum_{j\in\N} \delta_{\onevec\xi_j + U_{(j)}}
    ,
\end{equation*}
where $(\xi_j)_{j\in \N}$ are the points of a Poisson point process on $\R$ with intensity
$e^{-x} dx$ for $x \in \R$,
and $U_{(j)}$ are independent copies of $U$.
Moreover, we can represent the exponent measure as
\begin{equation}
    \label{eq:exponent_measure}
    \Lambda(A)
    =
    \int_{\R}
    e^{-x}
    \P(U + \onevec x \in A)
    dx
    ,
\end{equation}
for Borel sets $A \subseteq \Rd$.
The random vector $U$ can also serve as a $U$-generator of the multivariate generalized Pareto distribution in the sense of \citet[Proposition~9]{rootzen2018}.
Note that different generators $U$ can result in the same exponent measure $\Lambda$.
For more details see \cref{sec:ppp_construction} and \citet[][Section 2.1]{Corradini2024}.

Several popular parametric models for multivariate generalized Pareto distributions exist.

\begin{example}
    \label{ex:definitionLog}
    The extremal logistic model with parameter $\theta\in(0,1)$
    has generator $U = (U_1, \dots, U_d)$
    where $U_i$ are independent with Gumbel distribution with scale $\theta$
    and location $-\log\Gamma(1-\theta)$ \citep[e.g.,][]{dombry2016}.
\end{example}

\begin{example}
    \label{ex:definitionDirichlet}
    The Dirichlet model \citep{CT1991} with parameters $\alpha_1, \dots, \alpha_d>0$ has generator $U = (\log V_1, \dots, \log V_d)$
    where $V_i$ are independent $\text{Gamma}(\alpha_i, 1/\alpha_i)$ random variables
    \citep[e.g.,][]{kiriliouk2018POT,Corradini2024}.
\end{example}

\begin{example}
    \label{ex:definitionHR}
    The \HR{} generalized Pareto distribution \citep{hueslerReiss1989}
    is parameterized by a conditionally negative definite variogram matrix $\Gamma$ in the set $\Dcal$ of all symmetric matrices $\Gamma$ with zero diagonal and positive non-diagonal entries that satisfy $x\T\Gamma x < 0$ for all $x\in \onevec\orth$.
    Let
    $\Sigma = P_{\onevec} (-\half\Gamma) P_{\onevec}$,
    $\Theta = \Sigma\pinv$ be its pseudoinverse,
    and $r_\Theta = -\tfrac{1}{2d} \Theta\Gamma\onevec$.
    The \HR{} distribution has exponent measure density \citep{hentschel2023}
    \begin{equation}
         \label{eq:HR1}
        \lambda(y; \Gamma)
        \propto
        \exp\lr{
            -\half
            y\T \Theta y
            +
            r_\Theta\TT y
            -
            d\inv\onevec\T y
        }
        , \qquad
        y \in \Rd
        .
    \end{equation}
    For unit vectors $e_m$, $m = 1, \dots d$,
    similar expressions exist using
    $\Sigma^{e_m} := P_{e_m} (-\half\Gamma) P_{e_m}\T$
    \citep{engelke2015}.
    Corresponding to these density expressions,
    a \HR{} random vector $Y$ satisfies the stochastic representations
    \begin{equation}
        \label{eq:stoch_repr}
        Y \mid \set{Y_m > 0}
        \deq
        W\slr{m} + \onevec E
        , \qquad
        Y \mid \set{\onevec\T Y > 0}
        \deq
        W + \onevec E,
    \end{equation}
    where
    $W\slr{m}$ and $W$ are degenerate Gaussian vectors with covariance matrices
    $\Sigma^{e_m}$ and $\Sigma$, respectively,
    and $E \sim \Exp(1)$ is an independent exponential random variable.
\end{example}

\subsection{Variograms and Euclidean distance matrices}\label{SUBSEC:Euclidean}

Let $Z$ be a centered random vector taking values in $\Rd$ with finite covariance matrix $\Omega = \E\lr{Z Z\T}$.
Its variogram matrix $\Gamma$ is a symmetric $d\times d$-matrix given by the expected squared differences of the entries of $Z$, that is
\begin{equation}\label{vario_def}
    \Gamma_{ij}
    =
    \E\lr{Z_i - Z_j}^2
    =
    \Omega_{ii} + \Omega_{jj} - 2\Omega_{ij}
    .
\end{equation}
For any \pd{} or \psd{} $\Omega$ satisfying
$\kernel(\Omega) \cap \set{\onevec}\orth = \set{\zerovec}$,
the variogram matrix $\Gamma$ has full rank and is \cnd{}.
It is well known that the set of conditionally negative definite variogram matrices $\Dcal$ is equivalent to the set of Euclidean distance matrices \citep{gower1982}.
This means that if $\Gamma$ is conditionally negative definite there exists a set of $d$ points in $\R^{d-1}$,
which we denote as the row vectors $x_i$ of a matrix $X\in\Rd[d\times (d-1)]$,
such that for all $1 \leq i,j \leq d$ we have
\begin{equation*}
    \Gamma_{ij}
    =
    \norm{x_i - x_j}^2
    .
\end{equation*}
The covariance matrices $\Sigma$ and $\Sigma^{e_m}$, as defined in \cref{ex:definitionHR} above, are equal
to the inner products $XX\T$ of these points for certain choices of $X$.
As discussed for example in \citet{gower1982},
the (squared) distances $\Gamma_{ij}$ are invariant under translation of the points.
Let $\onevec\T v = 1$ and consider
$X^{v} = P_v X$,
which shifts the row vectors of $X$ by the same vector $v\T X$,
resulting in a point cloud such that the linear combination of points
given by $v$ is the origin, i.e., $v\T X = \zerovec\T \in \Rd[d-1]$.
The corresponding inner product matrix is
\begin{equation*}
    \Sigma^{v}
    =
    X^{v} (X^{v})\T
    .
\end{equation*}
In the context of a \HR{} distribution with variogram matrix $\Gamma$, a canonical unit vector $v = \kev{m}$
yields the covariance matrix $\Sigma^{e_m}$ of $W\slr{m}$
in \eqref{eq:stoch_repr}.
Similarly, choosing $v = \dinvvec$,
we get the covariance matrix $\Sigma$ of $W$ in \eqref{eq:stoch_repr}.
\begin{figure}
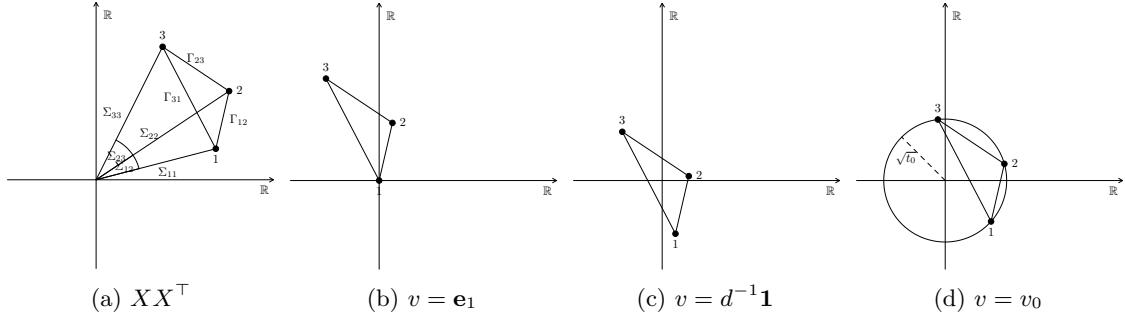

    \centering
    \begin{subfigure}{0.24\textwidth}
        \resizebox{\textwidth}{!}{%
            \inputtikzR{rowVectors/gamma_sigmaii_sigmaij}%
        }
        \caption{$XX\T$}
    \end{subfigure}
    \begin{subfigure}{0.24\textwidth}
        \resizebox{\textwidth}{!}{%
            \inputtikzR{rowVectors/shift_e1}%
        }
        \caption{$v = \kev{1}$}
    \end{subfigure}
    \begin{subfigure}{0.24\textwidth}
        \resizebox{\textwidth}{!}{%
            \inputtikzR{rowVectors/shift_onevec}%
        }
        \caption{$v = \dinvvec$}
    \end{subfigure}
    \begin{subfigure}{0.24\textwidth}
        \resizebox{\textwidth}{!}{%
            \inputtikzR{rowVectors/shift_circumhyper}%
        }
        \caption{$v = v_0$}
    \end{subfigure}
    \caption{%
        Initial matrix $\Sigma$ and the effect of shifting the row vectors of $X$, where the points correspond to row vectors of $X^v$.
    }
\end{figure}
Another interesting choice for $v$ is
\begin{equation}
    \label{eq:resistance}
    v_0 :=2t_0 \Gamma\inv\onevec,
    \qquad \text{where} \qquad
    t_0 := \frac{1}{2}(\onevec\T \Gamma\inv\onevec)\inv.
\end{equation}
In the literature, $t_0$ and $v_0$ are considered as resistance radius and resistance
curvature of the Euclidean distance matrix $\Gamma$,
as $v_0$ is a notion of discrete curvature,
see \cite{devriendt2022a}.
Clearly, the entries of $v_0$ sum up to one, because
$v_0\T \onevec = 2 t_0 \onevec\T \Gamma\inv \onevec = 1$.
The term resistance radius comes from the fact that this
value is the squared radius of the circumhypersphere that passes through all points
of the simplex that spans the Euclidean distance matrix $\Gamma$.

In the context of extremal graphical models \citep{engelke2020},
it is often convenient to parameterize
a \HR{} random vector $Y$ by the precision matrix
$\Theta = (P_{\onevec} (-\half\Gamma) P_{\onevec})\pinv$,
introduced in \cref{ex:definitionHR},
instead of the variogram matrix $\Gamma$.
This matrix is a signed graph Laplacian matrix
of the extremal conditional independence graph of $Y$ \citep{hentschel2023},
meaning that
$\Theta_{ij} = 0$ if and only if
$Y_i \perp_e Y_j \mid Y_{\set{1, \dots, d}\without{i,j}}$.
In this context, the matrix $\Theta$ is typically referred to as \HR{} precision matrix.
We show an example for a \HR{} graphical model in \cref{fig:HRGM}.
\begin{figure}
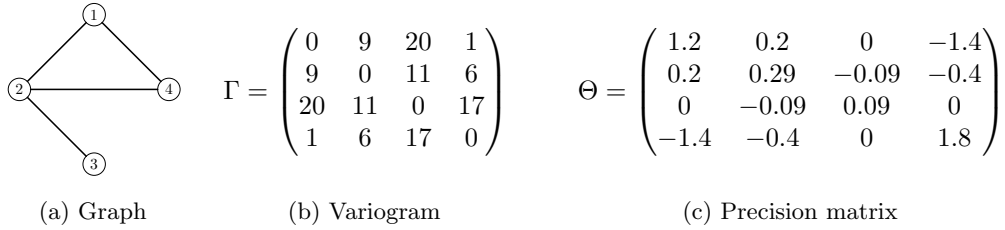

    \centering
    \begin{subfigure}{0.17\textwidth}
        \centering
        \resizebox{\textwidth}{!}{%
            \inputtikzR{graphs/G_03_graph}%
        }
        \caption{Graph}
    \end{subfigure}
    \begin{subfigure}{0.29\textwidth}
        \centering
        \begin{equation*}
            \Gamma
            =
            \mlr{\input{examples/Gamma_03.tex}}
        \end{equation*}
        \caption{Variogram}
    \end{subfigure}
    \begin{subfigure}{0.44\textwidth}
        \centering
        \begin{equation*}
            \Theta
            =
            \mlr{\input{examples/Theta_03.tex}}
        \end{equation*}
        \caption{Precision matrix}
    \end{subfigure}
    \caption{
        Variogram and precision matrices for a \HR{} graphical model with respect to the graph in (a).
    }
    \label{fig:HRGM}
\end{figure}

\section{Stochastic representations on non-negative half-spaces}
\label{sec:stoch_rep}

Throughout this section,
we consider vectors $u, v$ from the probability simplex $\Delta_{d-1}$
as defined in \cref{sec:linear_algebra}.
For such a vector $v$ and a multivariate generalized Pareto random vector $Y$ with standard exponential margins,
we denote the restriction of $Y$ to $\Hcal^v = \setm{x \in \Rd}{v\T x > 0}$
as $Y^v = Y \mid \set{v\T Y > 0}$.
In the special case $v = d\inv\onevec$,
we write $Y^\onevec$ instead of $Y^{d\inv\onevec}$
as this is more concise and mathematically equivalent.

\begin{remark}
    The requirement that $v \geq \zerovec$ ensures that the half-space $\Hcal^v$
    is contained in $\Lcal = \setm{x \in \Rd}{x \not\leq 0}$, the support of $Y$.
    Many of the results presented below can be extended to
    more general $v$, satisfying only $\onevec\T v = 1$,
    but we focus on the non-negative case to preserve
    a clear link between $Y^v$ and the original \multGP{} distribution $Y$.
\end{remark}

\subsection{Stochastic representation}

In this section we derive stochastic representations of the
random vectors $Y^v$ in terms of a so-called $v$-extremal function $W^v$.
We study the relation between versions of these random vectors for different
$v\in \Delta_{d-1}$.

\begin{proposition}\label{prop:extremefunctions}
    \provenlabel{prop:stochrep_halfspace}
    Let $Y$ be a
    \multGP{} random vector,
    $v \in \Delta_{d-1}$,
    and $Y^v = Y \mid \set{v\T Y > 0}$.
    Then we have
    the stochastic representation
    \begin{align*}
        Y^v
        &\deq
        W^v + E \onevec,
        \qquad\text{with}\\
        W^v
        &\deq
        P_v Y^v,
    \end{align*}
    where $E$ is a standard exponential random variable independent of $W^v$.
    We call $W^v$ the $v$-extremal function of $Y$ and have $\P(W^v \in v\orth) = 1$.
\end{proposition}

\begin{figure}
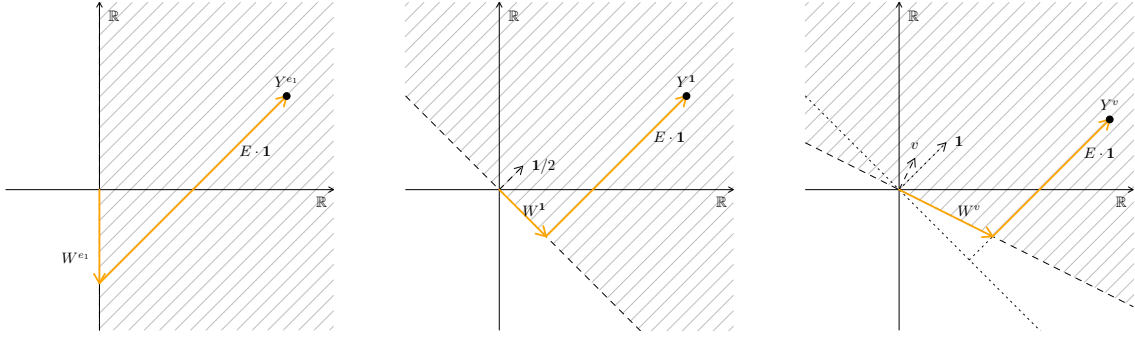

    \center
    \resizebox{0.30\textwidth}{!}{%
        \inputtikzR{densityProjections/e1}
    }
    \hfill
    \resizebox{0.30\textwidth}{!}{%
        \inputtikzR{densityProjections/OneVec}
    }
    \hfill
    \resizebox{0.30\textwidth}{!}{%
        \inputtikzR{densityProjections/V_onePerp}
    }
    \caption{%
        Geometric illustration of the stochastic representation in
        \cref{prop:stochrep_halfspace}
        for $v=e_1$ (left), $v=\onevec/2$ (center), and a general $v$ (right),
        including the extremal function $W^v$ and the corresponding multivariate generalized Pareto distributions $Y^v$.
    }
    \label{fig:densityProjection}
\end{figure}
Again, we use the more concise notation $W^\onevec = W^{d\inv\onevec}$.
This decomposition is illustrated in \cref{fig:densityProjection}.
The following results allow us to construct the $v$-extremal functions
and relate them to each other.
Recall the notion of a generator $U$ of a multivariate generalized Pareto distribution as defined in \cref{sec:MPD}.

\begin{proposition}
    \provenlabel{prop:generator_tilting}
    Let $U$ be any generator of $Y$ and
    let $W^v$ be the $v$-extremal function defined in \cref{prop:stochrep_halfspace}.
    We have the stochastic representation
    \begin{equation*}
        W^v
        \deq
        P_v U^v
        ,
    \end{equation*}
    where $U^v$ is the exponentially tilted random vector defined in distribution by
    \begin{equation*}
        \P(U^v \in A)
        =
        \frac{1}{
            \E\lr{\exp\lr{v\T U}}
        }
        \E\lr{
            \exp\lr{v\T U}
            \indicator{U \in A}
        }
        ,
    \end{equation*}
    for Borel sets $A \subseteq \Rd$.
    If $U$ has density $f_U$,
    then $U^v$ has density
    $
        f_{U^v}(y)
        =
        \frac{1}{
            \E\lr{\exp\lr{v\T U}}
        }
        \exp\lr{v\T y}
        f_U(y)
    $
    for $y \in \Rd$.
\end{proposition}
\begin{corollary}
    \provenlabel{cor:extremal_functions_are_generators}
    For any $v\in \Delta_{d-1}$ and $v$-extremal function $W^v$,
    the shifted random vector $W^v + \onevec c$
    with $c = \log \E\lr{\exp\lr{v\T U}}$
    is a generator of $Y$.
\end{corollary}

Since the distribution of $Y^v$ is defined by
the restriction of the exponent measure $\Lambda$ to the half-space $\Hcal^v$,
its density must be proportional to the exponent measure density $\lambda$,
if it exists.
In the proof of \cref{prop:generator_tilting},
the proportionality constant is computed
to be $\Lambda(\Hcal^v) = \E\lr{\exp\lr{v\T U}}$.
Note that this normalization constant is independent
of the chosen generator $U$ of a given multivariate generalized Pareto distribution $Y$.
We have the following representation of the density of $Y^v$.

\begin{corollary}
    \provenlabel{cor:density_Yv_from_lambda}
    If $Y$ has an absolutely continuous exponent measure $\Lambda$ with density $\lambda$,
    then for $v \in \Delta_{d-1}$ the density of $Y^v$ is
    \begin{equation*}
        f_{Y^v}(y)
        =
        \frac{1}{
            \E\lr{\exp\lr{v\T U}}
        }
        \indicator{v\T y > 0}
        \lambda(y)
        .
    \end{equation*}
    For $u \in \Delta_{d-1}$ and for $y$ satisfying both $v\T y > 0$ and $u\T y > 0$,
    we have
    \[ 
        f_{Y^v}(y)
        = f_{Y^u}(y)
        \frac{
            \E\lr{\exp\lr{u\T U}}
        }{
            \E\lr{\exp\lr{v\T U}}
        }
        . 
    \]
\end{corollary}

\cref{prop:generator_tilting} allows us to compute $v$-extremal functions
if a generator $U$ of $Y$ is known.
Furthermore,
we can use the additivity of exponential tilting
to obtain the following relationship between different extremal functions.
\begin{lemma}
    \provenlabel{corr:extremal_function_tilting}
    Consider two vectors $u,v \in \Delta_{d-1}$
    and the corresponding extremal functions $W^u$ and $W^v$.
    Then, for any Borel set $A^v \subset v\orth$,
    and with $A^u = P_u A^v \subset u\orth$,
    we have
    \begin{equation*}
        \P(W^v \in A^v)
        =
        \frac{1}{
            \E\lr{\exp\lr{v\T W^u}}
        }
        \E\lr{
            \exp\lr{v\T W^u}
            \indicator{W^u \in A^u}
        }
        .
    \end{equation*}
    If $W^u$ has density $f_{W^u}$ on $u\orth$,
    then $W^v$ has the following density
    on $v\orth$:
    \begin{equation*}
        f_{W^v}(y)
        =
        \frac{
            \norm{u}
        }{
            \norm{v}
        }
        \frac{1}{
            \E\lr{\exp\lr{v\T W^u}}
        }
        \exp\lr{-u\T y}
        f_{W^u}(P_u y)
        .
    \end{equation*}
\end{lemma}
\begin{remark}
    \label{remark:hyperplaneMeasures}
    Throughout the paper,
    we consider densities of random vectors supported on hyperplanes $v\orth$
    with respect to the $(d-1)$-dimensional Hausdorff measure on these hyperplanes,
    that is,
    the pushforward measure obtained by
    the $(d-1)$-dimensional Lebesgue measure on $\Rd[d-1]$
    under any orthonormal transformation from $\Rd[d-1]$ to $v\orth$.
    See \cref{sec:density_details} for more details.
\end{remark}

Since $Y^v$ is the sum of independent random vectors $W^v$ and $E\onevec$,
which are supported on complementary
subspaces,
its density can be computed as the product of the densities of $W^v$ and $E\onevec$,
correcting for the Jacobian of the transformation.

\begin{proposition}
    \provenlabel{prop:density_Yv}
    If $W^v$ has density $f_{W^v}$ on $v\orth$,
    then the density of $Y^v$ is given by
    \begin{equation*}
        f_{Y^v}(y)
        =
        f_{W^v}(P_v y)
        \norm{v}
        \exp\lr{-v\T y}
        , \quad
        y \in \Hcal^v
        .
    \end{equation*}
\end{proposition}

\subsection{The \texorpdfstring{$v$}{v}-variogram}

The variogram is an important summary statistic of a random vector.
We introduce the $v$-variogram that measures dependence in the vector $Y^v$.

\begin{definition}
    \label{def:v_variogram}
    For a \multGP{} random vector $Y$ and a vector $v \in \Delta_{d-1}$,
    writing $Y^v = Y \mid \set{v\T Y > 0}$,
    we define the $v$-variogram $\Gamma^v$ as the $d \times d$-matrix with entries
    \begin{equation*}
        \Gamma_{ij}^{v}
        =
        \Var\lr{Y_i^v - Y_j^v}
        ,
    \end{equation*}
    provided the variances exist and are finite.
\end{definition}

Since we have the relations $P_v Y^v \deq W^v \deq P_v U^v$,
the random vectors $Y^v$, $W^v$, and $U^v$ differ in distribution only
by an additive (random) value along the $\onevec$-direction,
which cancels out in the computation of the $v$-variogram $\Gamma^v$,
yielding the following result.

\begin{corollary}\label{lem:variogram_formulas}
    \provenlabel{cor:v_variogram_expression}
    For $Y^v$, $W^v$ and $U^v$ as above,
    we have
    \begin{align*}
        \Gamma_{ij}^{v}
        &=
        \Var\lr{Y_i^v - Y_j^v}
        \\ &=
        \Var\lr{W_i^v - W_j^v}
        \\ &=
        \Var\lr{U_i^v - U_j^v}
        .
    \end{align*}
\end{corollary}

Since many models are given through simple forms of their generators $U$,
\cref{cor:v_variogram_expression} together with \cref{prop:generator_tilting}
allow us to compute $\Gamma^v$ for the three examples introduced in \cref{sec:MPD}.
Let $\psi^{(1)}$ denote the trigamma function.

\begin{example}
    \provenlabel{ex:variogram_logistic}
    For the logistic model (\cref{{ex:definitionLog}}) with parameter $\theta \in (0,1)$,
    we have for $i \neq j$
    \begin{equation*}
        \Gamma_{ij}^v
        =
        \theta^2
        \psi^{(1)}\lr{1 - v_i \theta}
        +
        \theta^2
        \psi^{(1)}\lr{1 - v_j \theta}
        .
    \end{equation*}
\end{example}

\begin{example}
    \provenlabel{ex:variogram_dirichlet}
    For the Dirichlet model (\cref{ex:definitionDirichlet}) with parameters $\alpha_1, \dots, \alpha_d > 0$,
    we have for $i \neq j$
    \begin{equation*}
        \Gamma_{ij}^{v}
        =
        \psi^{(1)}(\alpha_i + v_i)
        +
        \psi^{(1)}(\alpha_j + v_j)
        .
    \end{equation*}
\end{example}

Note that for the $m$-th canonical unit vector $v = e_m$,
\cref{ex:variogram_logistic,ex:variogram_dirichlet} recover Examples~2 and~3 from \cite{engelkeVolgushev2022}.

\begin{example}\provenlabel{ex:HR}
    For a covariance matrix $\Sigma$,
    let $U \sim \normal\lr{\mu, \Sigma}$,
    with $\mu = -\half d_\Sigma$ and
    $d_\Sigma$ denoting the vector of diagonal entries of $\Sigma$.
    This generator yields the \HR{} generalized Pareto distribution in \cref{ex:definitionHR} parametrized by variogram
    $\Gamma = d_\Sigma\onevec\T + \onevec d_\Sigma\T - 2\Sigma$.
    Since both exponential tilting and projection preserve Gaussianity,
    the $v$-extremal functions can be computed to be
    \begin{equation}
        \label{eq:HR_extremal_function}
        W^v
        \deq
        P_v U^v
        \sim
        \normal\lr{\mu_v, \Sigma^v}
        ,
    \end{equation}
    with mean and covariance uniquely determined by $\Gamma$ as
    \begin{equation}
    \label{eq:muvSigv}
        \mu_v
        =
        P_v \mu + P_v \Sigma v
        =
        -\half
        P_v \Gamma v
        , \qquad
        \Sigma^v
        =
        P_v \Sigma P_v\T
        =
        P_v (-\half\Gamma) P_v\T
        .
    \end{equation}
    Since the variogram only depends on the covariance matrix and is invariant to shifts along the $\onevec$-direction,
    we find that the $v$-variogram $\Gamma^v$ is independent of $v$ and coincides with the parameter matrix $\Gamma$.
    By \cref{lem:variogram_formulas} we have
    \begin{equation*}
        \Gamma_{ij}^v
        =
        \Var(U_i^v - U_j^v)
        =
        \Sigma^v_{ii} + \Sigma^v_{jj} - 2\Sigma^v_{ij}
        =
        \Gamma_{ij}
        .
    \end{equation*}
\end{example}

\section{\HR{} density via noncanonical half-spaces}
\label{sec:density}

Among multivariate generalized Pareto distributions,
the parametric \HR{} family introduced in \cref{ex:definitionHR} has many special properties that are interesting for statistical modeling and inference.
For example, \HR{} models are parameterized by a variogram matrix $\Gamma$, which coincides with their $v$-variogram matrices for any $v\in \Lambda_d$.
Furthermore, the \HR{} precision matrix $\Theta$ permits parsimonious graphical modeling in extremes, which has inspired a recent series of research papers; see \citet{engelke2024} for an overview.

Throughout this section, let $Y$ be a \HR{} random vector with variogram matrix $\Gamma$.
As in \cref{ex:definitionHR},
we write $\Sigma = P_{\onevec} (-\half \Gamma) P_{\onevec}$ and
$\Theta = \Sigma\pinv$.
Again consider vectors $u,v \in \Delta_{d-1}$ and let related objects
such as $Y^u$ and $W^v$ be defined as in \cref{sec:stoch_rep}.

\subsection{Exponent measure density expression}

Using \cref{prop:density_Yv} and the density of the degenerate multivariate normal distribution in \cref{eq:HR_extremal_function},
we can directly express the density of $Y^v$ as
\begin{equation}
    \label{eq:density_HR_Yv}
    f_{Y^v}(y)
    =
    \sqrt{(2\pi)^{-(d-1)}\pdet{\Sigma^v}\inv}
    \exp\lr{
        -\half\norm[(\Sigma^v)\pinv]{P_v y-\mu_v}^2
    }
    \norm{v}
    \exp\lr{-v\T y}
    ,
\end{equation}
with $\Sigma^v$ and $\mu_v$ as in \cref{ex:HR},
$\norm[A]{y}^2 := y\T A y$,
and
$\pdet{\cdot}$ denoting the pseudo-determinant,
that is,
the product of the non-zero eigenvalues.
See \cref{lemma:hyperplane_normal_density} for a derivation of the
density of the degenerate normal distribution on the hyperplane $v\orth$.
However, in order to express the exponent measure density $\lambda$ similarly,
it remains to compute the normalization constant $\Lambda\lr{\Hcal^v}$.
Furthermore, it turns out that the projection $P_v y$ and subsequent multiplication with $(\Sigma^v)\pinv$ in the quadratic form
naturally simplifies to multiplication with $\Theta$,
and the pseudo-determinant $\pdet{\Sigma^v}$ can be expressed in terms of $\pdet{\Theta}$.
In the following \namecref{lemma:integralHyper},
we provide a closed formula for $\Lambda(\Hcal^v)$
for general $v\in \Rd$ satisfying $\onevec\T v=1$.
For $v\in\Delta_{d-1}$, this gives the normalizing constant of the probability density in \cref{cor:density_Yv_from_lambda} for the \HR{} distribution.
\begin{lemma}
    \provenlabel{lemma:integralHyper}
    Let $Y$ follow a \HR{} distribution with variogram matrix $\Gamma$ and let $v\in \Rd$ with $\onevec\T v=1$.
    Then
    \begin{equation*}
        \Lambda\lr{\Hcal^v}
        =
        \E\lr{\exp\lr{v\T U}}
        =
        \exp\lr{
            -\tfrac{1}{4}v\T\Gamma v
        }
        .
    \end{equation*}
\end{lemma}
For canonical unit vectors $v=e_m$, we recover the identity $\Lambda\lr{\Hcal^{e_m}}=1$.
For noncanonical $v \in \Delta_{d-1}$,
the value of $\Lambda\lr{\Hcal^v}$ is strictly smaller than one,
since $\Gamma$ has strictly positive entries in off-diagonal entries.
For $v=\dinvvec$, the exponent is proportional to the Kirchhoff index $\frac{1}{2}\onevec\T \Gamma \onevec = d\, \trace(\Sigma)$,
a scalar graph invariant, which is relevant for example in chemical graph theory \citep{devriendt2022a}.

As shown in \cref{ex:HR}, the extremal functions $W^v$
have equivalent covariance matrices in the sense that they are all projections $\Sigma^v = P_v\T(-\half\Gamma) P_v$ of the same matrix.
In order to relate their pseudo-determinants to each other and simplify the corresponding density expressions,
we give the following result.
\begin{lemma}
    \provenlabel{lem:pseudoDeterminants}
    For vectors $u,v$ satisfying $\onevec\T u = \onevec\T v = 1$,
    the pseudo-determinants of $\Sigma^u, \Sigma^v$ from \cref{ex:HR} satisfy
    \begin{equation*}
        \frac{
            \pdet{\Sigma^{u}}
        }{
            \pdet{\Sigma^{v}}
        }
        =
        \frac{\norm{u}^2}{\norm{v}^2}
        .
    \end{equation*}
\end{lemma}
Plugging in $u = \dinvvec$,
this implies $\pdet{\Sigma^v}=d\norm{v}^2\pdet{\Sigma}$,
and using $u = e_m$
it recovers
$\det{\Sigma\slr{m}} = d\pdet{\Sigma}$
for the $(d-1)$-dimensional covariance matrices
$\Sigma\slr{m} = \Sigma^{e_m}_{\setminus{m},\setminus{m}}$
studied for example in \cite{engelke2020}.
Combining these results with those from \cref{sec:stoch_rep},
we derive the following closed form expressions for the densities $\lambda$ and $f_{Y^v}$.

\begin{proposition}
    \provenlabel{prop:HrDensity}
    Let $Y$ follow a \HR{} distribution with variogram matrix $\Gamma$ and $v \in \Delta_{d-1}$.
    The density of $Y^v$ for $y \in \Hcal^v$ can be expressed as
    \begin{equation}
        \label{eq:extremal_density_HR}
        f_{Y^v}(y)
        =
        \sqrt{d\inv(2\pi)^{-(d-1)}\pdet{\Theta}}
        \exp\lr{
            -\half\norm[\Theta]{y-\tilde\mu_v}^2
        }
        \exp\lr{-v\T y}
        ,
    \end{equation}
    with $\tilde\mu_v := P_{\onevec}\mu_v = P_{\onevec}\lr{-\half\Gamma}v$.
    For the exponent measure density $\lambda$, we have
    \begin{equation}
        \label{eq:HR2}
        \lambda(y)
        =
        \exp\lr{-\tfrac{1}{4}v\T\Gamma v}
        f_{Y^v}(y)
    \end{equation}
    for any $y \in \Rd$ and $v\in \Rd$ with $\onevec\T v = 1$,
    using the algebraic expression from \cref{eq:extremal_density_HR}
    for values of $y$ outside the support of $f_{Y^v}$
    or $v \notin \Delta_{d-1}$.
\end{proposition}

Note that the value of the right-hand side in \cref{eq:HR2} does not depend on $v \in \Rd$,
as long as $\onevec\T v = 1$,
which is why $v$ is omitted on the left-hand side.
Using the general expression of $f_{Y^v}$ in \cref{eq:density_HR_Yv} and \cref{prop:HrDensity},
we can recover some well-known density representations
arising from particular choices of $v$.

\begin{example}
    Let $v=e_m$ be a canonical unit vector and
    let $\mu_{e_m}$ and $\Sigma^{e_m}$ be as in \cref{eq:muvSigv}.
    Simplifying \cref{eq:density_HR_Yv} then yields
    \begin{equation*}
        f_{Y^{e_m}}(y)
        =
        \sqrt{
            (2\pi)^{-(d-1)}
            \pdet{\Sigma^{e_m}}\inv
        }
        \exp\lr{
            -\half\norm[\Theta_{e_m}]{y-\mu_{e_m}}^2
            - y_m
        }
        .
    \end{equation*}
\end{example}
Observing that $e_m\T \Gamma e_m = 0$,
\cref{eq:HR2} furthermore implies that the same expression gives the exponent measure density $\lambda(y)$.
After transforming to standard Pareto margins, this recovers the exponent measure density expression
used for example in expression~(9) of \cite{engelke2020}.
Setting $v = \dinvvec$ recovers a similar
expression to \cref{eq:HR1} and
Proposition~3.4 in \cite{hentschel2023}.
\begin{example}
    Let $v=\dinvvec$ and
    observe that $\mu_{\onevec}$ from \cref{eq:muvSigv} and $\tilde\mu_{\onevec}$
    from \cref{prop:HrDensity} coincide for this choice of $v$.
    Then \cref{eq:extremal_density_HR} yields
    \begin{equation*}
        f_{Y^{\onevec}}(y)
        =
        \sqrt{
            d\inv
            (2\pi)^{-(d-1)}
            \pdet{\Theta}
        }
        \exp\lr{
            -\half\norm[\Theta]{y-\mu_{\onevec}}^2
            - \dinvvec\T y
        }
        .
    \end{equation*}
\end{example}

\subsection{The resistance curvature vector}
\label{sec:special_halfspace}

In the exponent measure density expression in \cref{prop:HrDensity}, the vector $v\in\Rd$ can be any vector that satisfies $\onevec\T v=1$.
Standard choices include the canonical unit vectors and the vector $\dinvvec$.
In this section we will show that the resistance curvature vector $v_0$ introduced in \cref{eq:resistance} leads to a particularly elegant exponent measure density expression.
Our simulation study in \cref{sec:inferenceUsingV} gives evidence for interesting statistical properties for this choice.

Let $\Gamma\in \Dcal$ be a variogram
and define the matrix  $M_t = t\onevec\onevec\T - \half\Gamma$
for $t \in \R$.
By \citet[][Proposition~3.2]{hentschel2023},
$M_t$ is invertible for all $t \neq t_0$,
where $t_0 = \half(\onevec\T\Gamma\onevec)\inv$
is the resistance radius of $\Gamma$.
Furthermore, $M_t$ is connected to the \HR{} precision matrix $\Theta$
through the relationship
$\Theta = \lim_{t \to \infty} M_t\inv$.
We observe that the resistance curvature $v_0$ satisfies
\begin{equation*}
    M_{t_0}v_0
    =
    (t_0\onevec\onevec\T-\half \Gamma)v_0=t_0\onevec - t_0\onevec = \zerovec
    ,
\end{equation*}
implying that $v_0$ spans the one-dimensional kernel of $M_{t_0}$.
In the following result we find that the vector $v_0$, if chosen as the defining vector in \cref{prop:HrDensity}, leads to a particularly simple representation of the \HR{} exponent measure density.
It has already been observed to play a special role in optimal prediction
in \HR{} models where it naturally appears in the kriging formula \citep[][Proposition 5.4]{bolin2025}.
Furthermore, we obtain a particularly simple stochastic representation of the \HR{} distribution restricted to the half-space spanned by $v_0$.

\begin{corollary}\label{PROP:special_halfspace}
    \provenlabel{prop:HRdensityLimit}
    The exponent measure density $\lambda$
    of a \HR{} distribution with variogram $\Gamma$ and precision matrix $\Theta$
    can be expressed as
    \begin{equation*}
        \lambda(y)
        =
        \sqrt{
            d\inv
            \lr{2\pi}^{-(d-1)}\pdet{\Theta}
        }
        \cdot
        \exp\lr{-\half t_0}
        \cdot
        \exp\lr{
            -\half y\T \Theta y
            - v_0\T y
        }
        , \quad
        y \in \Rd
        .
    \end{equation*}
    If $v_0\in\Delta_{d-1}$ it holds that
    $W^{v_0}\sim \normal(\zerovec,\Sigma^{v_0})$
    with
    $\Sigma^{v_0} = t_0\onevec\onevec\T-\half\Gamma$.
\end{corollary}
In particular, this means that if $v_0 \in \Delta_{d-1}$, the $v_0$-extremal function $W^{v_0}$
is a centered Gaussian vector.
The representation of $\lambda$ allows rewriting the exponent into a quadratic form involving $t_0,v_0$ and $\Theta$ \citep[Example~5]{EGR2025}.
Since, by construction, $\Sigma^{v_0}$ has a constant diagonal,
it can be standardized to a correlation matrix.
This gives rise to a standardized variogram and an analog of the (Gaussian) elliptope, that is, the bounded set of all positive semi-definite correlation matrices, see \citet{DER2026}.
Furthermore, we find that the half-space spanned by $v_0$ has the smallest volume
with respect to the exponent measure among all vectors $v$ that sum up to one.
\begin{corollary}\label{cor:minimal_halfspace}
    \provenlabel{cor:minIntegral}
    For $v\in \Delta_{d-1}$ and $m=1,\dots,d$ we have
    \begin{equation*}
        \exp\lr{-\half t_0}
        =
        \Lambda\lr{\Hcal^{v_0}}
        \leq
        \Lambda\lr{\Hcal^v}
        \leq
        \Lambda\lr{\Hcal^{e_m}}
        =
        1
        .
    \end{equation*}
    Furthermore, the first inequality also holds for all $v\in\Rd$ with $\onevec\T v=1$ and is strict for $v \neq v_0$.
    The second inequality is strict for any $v \in \Delta_{d-1}$ that is not a canonical unit vector.
\end{corollary}

This result has important consequences for statistical inference,
since it states that the half-space $\Hcal^{v_0}$ contains in expectation the smallest number of observations of $Y$.
For the empirical version $\hat\Gamma^v$ of $\Gamma^v$,
the variance of
$\hat \Gamma^{v_0}$ is thus expected to be
the largest across all $\hat \Gamma^v$ for $v\in\Delta_{d-1}$.
On the other hand, as we will see in the simulation study in \cref{sec:inferenceUsingV}, $\hat\Gamma^{v_0}$ often has a much smaller bias.

\begin{remark}
    In the context of statistical inference,
    we only consider
    vectors $v \in \Delta_{d-1}$,
    since only these define half-spaces $\Hcal^v$ contained in the support $\Lcal = \setm{x \in \Rd}{x \not\leq \zerovec}$
    of the \multGP{} distribution $Y$.
    In the case that $v_0$ has negative entries,
    we therefore consider an approximation $v_0^* \in \Delta_{d-1}$
    which is closest to $v_0$ in the sense that it defines a half-space $\Hcal^{v_0^*} \subseteq \Lcal$
    with a minimal value of the \HR{} exponent measure $\Lambda\lr{\Hcal^{v_0^*}}$, that is,
    \begin{equation}
        \label{eq:v0*}
        v_0^*
        =
        \argmin_{v \in \Delta_{d-1}}
        \Lambda\lr{\Hcal^v}
        =
        \argmin_{v \in \Delta_{d-1}}
        -v\T \Gamma v
        .
    \end{equation}
    The second equality follows from \cref{lemma:integralHyper}.
    Since $\Gamma$ is a \cnd{} matrix
    and $v$ is required to satisfy $\onevec\T v = 1$,
    this optimization problem is convex and thus efficiently solvable using standard solvers.
    By \cref{cor:minIntegral}, $v_0^*$ agrees with $v_0$ if the latter is non-negative.
\end{remark}

\section{Inference based on noncanonical half-spaces}
\label{sec:inferenceUsingV}

In this section we introduce an estimator for $v$-variograms and provide an extensive simulation study that illustrates its application and performance.

\subsection{The empirical \texorpdfstring{$v$}{v}-variogram}
\label{sec:empvario}
In the following definition we generalize the empirical variogram of \citet{engelkeVolgushev2022} to an estimator of the $v$-variogram as introduced in \cref{def:v_variogram}.
We further introduce an ensemble estimator for collections of vectors in $\Delta_{d-1}$.
\begin{definition}
    \label{def:empirical_variogram}
    Let $Y\slr{1},\ldots,Y\slr{n}$
    be i.i.d.~copies of a multivariate generalized Pareto vector $Y$.
    We define the empirical $v$-variogram as
    $\hat\Gamma^v \in \Rdd$ with entries
    \begin{align*}
        \hat\Gamma^v_{ij}
        &=
        \widehat{\Var}
        \lr{
            Y_i - Y_j \mid v\T Y > 0
        },
        \\ &=
        \frac{1}{\abs{\Jcal^v}-1}
        \sum_{k\in\Jcal^v}
        \lr{
            (Y\slr{k}_i - Y\slr{k}_j)
            -
            (\overline Y_i - \overline Y_j)
        }^2
        ,
    \end{align*}
    where
    $\Jcal^v = \setm{j = 1,\ldots,n}{v\T Y\slr{j} > 0}$,
    and
    $\overline Y_i = \abs{\Jcal^v}\inv \sum_{k\in\Jcal} Y\slr{k}_i$.
    For a finite, non-empty set of vectors $V \subset \Delta_{d-1}$,
    define the ensemble estimator
    \begin{equation}
        \label{eq:ensembleEstimator}
        \hat\Gamma^V
        =
        \frac{1}{\abs{V}}\sum_{v\in V} \hat\Gamma^v
        .
    \end{equation}
\end{definition}
Since $\hat\Gamma^v_{ij}$ is the sample version of the true variance $\Var(Y_i^v - Y_j^v)$,
we immediately obtain consistency and unbiasedness of these estimators as long as the variance $\Gamma^v_{ij}$ exists.

In practice, we cannot directly observe data from a multivariate generalized Pareto distribution.
Instead, we apply the following steps in order to obtain an approximate sample for the multivariate generalized Pareto distribution, see e.g.~\citet[Section~7.1]{roettger2023}.
Let $\tilde{X}$ be a data-generating process with continuous marginal cumulative distribution functions $F_1,\ldots,F_d$.
Then, the transformed random vector $X$ with $X_i=-\log(1-F_i(\tilde{X}_i))$ has standard exponential margins.
We assume that $X$ satisfies the limit \eqref{eq:mpd} for some multivariate generalized Pareto vector $Y$,
or, equivalently, that $X$ lies in the domain of attraction of $Y$.

Let $\tilde{x}\in \R^{n\times d}$ be a data matrix of $n$ i.i.d.~observations of $\tilde{X}$.
In order to obtain a matrix $x\in \R^{n\times d}$ with approximate observations of $X$, it is common in extremal dependence modeling to employ empirical cumulative distribution functions, see e.g.~\citet{ES2009}.
For each $i=1,\ldots,d$, let $\hat{F}_i$ be the empirical distribution function of
$\tilde{x}_{i1},\ldots,\tilde{x}_{in}$,
that is of the $i$-th column of $\tilde{x}$.
We can now obtain an approximate data matrix
$x\in\R^{n\times d}$
via transformations
$x_{ji}=-\log(1-\frac{n}{n+1}\hat{F}_i(\tilde{x}_{ji}))$
for all $j=1,\ldots,n$ and $i=1,\ldots,d$.
Finally, we threshold and standardize the observations in $x$ to construct an approximate sample for the limiting multivariate generalized Pareto vector $Y$.
To this end, we select a high quantile of the standard exponential distribution $q_p$ for $p \in (0,1)$ close to one, and define
\begin{equation*}
    y_j
    =
    x_j - q_p\onevec
    \quad \forall
    j \in \Ical
    =
    \setm{\ell = 1,\ldots,n}{\max{x_\ell}>q_p}
    .
\end{equation*}
The resulting data matrix $y \in \Rd[m\times d]$, for some (random) $ m = \abs{\Ical} \le n$,
is now considered as an approximate sample for the multivariate generalized Pareto vector $Y$.
This permits the approximate calculation of the empirical variogram $\hat\Gamma^v$ as in \Cref{def:empirical_variogram} for any $v\in \Delta_{d-1}$.

\subsection{Simulation setup}

In the remainder of this section we present the results of a simulation study to explore how the empirical $v$-variogram from \cref{def:empirical_variogram}
can be used for statistical inference and how the choice of $v$ affects the resulting estimator $\hat\Gamma^v$.
We focus on the \HR{} distribution,
for which all $v$-variograms $\Gamma^v$ coincide with the parameter matrix $\Gamma$,
independently of $v \in \Delta_{d-1}$.

For a given dimension $d$ and variogram $\Gamma$, we consider two distributions:
The limiting \multGP{} distribution with exponential margins as defined in \cref{ex:definitionHR},
and a distribution in the domain of attraction of this distribution in the sense of \cref{eq:mpd}, from which we obtain an approximate sample as described in \cref{sec:empvario} above.
There are many such distributions, and we use the max-stable \HR{} distribution
with the same variogram~$\Gamma$ since it is easy to simulate from.
For the limiting distribution, we directly apply the definition of $\hat\Gamma^v$ from \cref{def:empirical_variogram},
whereas for the distribution in the domain of attraction,
we apply the definition of the empirical variogram to the approximate sample obtained after thresholding and transformation.

Throughout, we consider different choices of the dimension~$d$,
number of observations~$n$, and threshold quantile~$p$ for the distribution in the domain of attraction.
To ensure comparability of the results,
the same randomly generated variogram $\Gamma$ is considered in all setups with the same dimension $d$.
In order to estimate the bias and variance of the estimators $\hat\Gamma^v$,
we generate $M=2000$ datasets for each setup.

The \texttt{R} package \texttt{graphicalExtremes} \citep{graphicalExtremes2024} is used to sample the variograms $\Gamma$,
generate data from the considered distributions, and perform the thresholding and transformation procedure.

\subsection{Two-dimensional case}

First, we perform a simulation study for the two-dimensional case.
This setup is particularly easy to visualize,
since any two-dimensional variogram $\Gamma$ is uniquely parametrized by its
off-diagonal entry $\gamma := \Gamma_{12} = \Gamma_{21}$,
and the allowed directions $v$ can be parametrized
by a single value $\eta \in \brackets{0, 1}$ as $v = (\eta, 1-\eta) \in \Delta_2$.
Due to symmetry, we have $v_0 = v_0^* = (0.5, 0.5)$ for any value of $\gamma$.
We consider regularly spaced values of $\eta \in \brackets{0, 1}$
and use $\gamma = 1.5$ as the true value for the variogram parameter.

\cref{figure:gammaEmp2d} illustrates the performance of the corresponding estimators $\hat\Gamma^v$.
Furthermore, we consider the ensemble estimator defined in \cref{eq:ensembleEstimator},
based on $V = \set{e_1, e_2}$,
equivalent to the empirical variogram from \cite{engelkeVolgushev2022},
and $V$ consisting of the evenly spaced vectors $v$ considered above.
\cref{table:VBE_joint} shows the biases, variances, and RMSEs
of the resulting estimators for $e_1$, $v_0$, and the two ensembles.

\begin{figure}
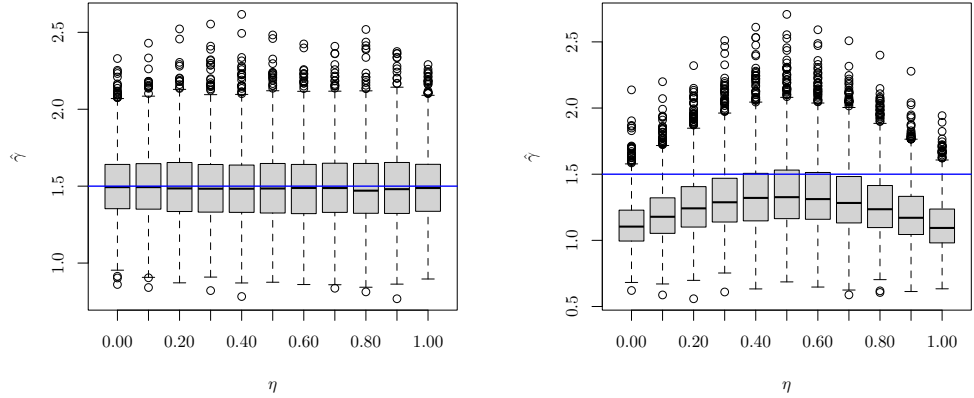

    \centering
    \resizebox{0.9\textwidth}{!}{%
        \inputtikzR{simulation2d/gammaEmp_rmpareto_estimate}
        \inputtikzR{simulation2d/gammaEmp_rmstable_estimate}
    }
    \caption{%
        Performance of estimators $\hat\Gamma^v$
        for different choices of $v$ and $M=2000$ dataset realizations,
        based on the limiting \HR{} distribution (left)
        as well as
        thresholded data in the domain of attraction (right).
        The true value of $\gamma = \Gamma_{12}$ is indicated by the
        blue line.
    }
    \label{figure:gammaEmp2d}
\end{figure}

\begin{table}
    \centering
    \resizebox{\textwidth}{!}{%
        
  \begin{tabular}{|l|rrrr|rrrr|}
      \hline
      Model & \multicolumn{4}{c|}{Limiting Model} & \multicolumn{4}{c|}{Domain of Attraction} \\
      \hline
      Vector(s) & $e_1$ & $v_0$ & $\{e_1, e_2\}$ & $\{v_1, \dots, v_{11}\}$ & $e_1$ & $v_0$ & $\{e_1, e_2\}$ & $\{v_1, \dots, v_{11}\}$ \\
      \hline
      Variance & 0.047 & 0.057 & 0.025 & 0.032 & 0.033 & 0.082 & 0.027 & 0.049 \\
      Bias & 0.004 & -0.004 & 0.000 & -0.003 & -0.380 & -0.140 & -0.382 & -0.241 \\
      RMSE & 0.216 & 0.239 & 0.157 & 0.178 & 0.421 & 0.319 & 0.416 & 0.327 \\
      \hline
  \end{tabular}

    }
    \caption{%
        Bias, variance, RMSE of different estimators.
        Sets of vectors denote ensemble estimators which average
        over the corresponding $v$-variograms.
    }
    \label{table:VBE_joint}
\end{table}

In the limiting model,
we observe that all estimators $\hat\Gamma^v$ are unbiased and differ only in their variance.
In fact, the different variances can be perfectly explained by the different half-space measures $\Lambda(\Hcal^v)$,
which are proportional to the expected number of observations used to compute $\hat\Gamma^v$
for each $v$.
For the ensemble estimators we observe a significant reduction in variance.

For the thresholded data in the domain of attraction, we observe a more complex picture.
Vectors near the unit vectors $e_m$ yield estimators with low variance but high bias,
whereas vectors near $v_0^*$ yield estimators with low bias and high variance.
Considering ensembles of vectors reduces the variance as expected.

This illustrates a classical issue in multivariate extreme value theory.
For data from a random vector $X$ in the domain of attraction of a multivariate generalized Pareto distribution $Y$,
components $X_j$ that do not exceed the threshold might still be far away (in distribution) from their limit.
One way to mitigate the resulting estimation bias is to perform censoring \citep[e.g.,][]{wadsworthTawn2014},
which however becomes computationally very costly even in moderate dimensions.
An intuitive explanation for the lower bias of the empirical variogram with vector $v_0$
is that it selects points in lower-density regions of the exponent measure,
which are therefore more extreme.

\subsection{Inference based on different half-spaces}

The simulation setup in higher dimensions is more complex and harder to visualize
than in the two-dimensional case.
We consider dimension $d=10$ and a fixed, randomly generated variogram $\Gamma \in \Rdd$ as true parameter.
The resistance vector $v_0$ of this variogram $\Gamma$ does not satisfy $v_0 \geq \zerovec$, and therefore we use the adapted vector $v_0^*$ in \cref{eq:v0*} instead.

In order to consider not just vectors from the interior of the simplex $\Delta_{d-1}$,
we enforce different sparsity levels for the vectors $v$,
by randomly setting $0, 1, \dots, d-2$ entries of $v$ to zero.
The remaining non-zero entries are then sampled
uniformly from the corresponding face of the probability simplex.
Furthermore, we consider the
vectors $v_0^*$, $\dinvvec$, and the unit vectors $e_m$ for $m=1, \dots, d$.
In total, we consider the same 500 distinct vectors for each setup.

To evaluate the performance of different estimators $\hat\Gamma^v$,
we compute their (average) squared bias and variance as
\begin{align*}
    \text{Bias}^2(\hat\Gamma^v)
    &=
    \frac{1}{d(d-1)}\sum_{i\neq j}
    \lr{\E{\hat\Gamma^v_{ij}} - \Gamma_{ij}}^2,
    \\
    \text{Variance}(\hat\Gamma^v)
    &=
    \frac{1}{d(d-1)}\sum_{i\neq j}
    \Var\lr{\hat\Gamma^v_{ij}},
\end{align*}
replacing expectations and variances by their sample versions based on the $M$ dataset realizations.
It turns out that the half-space measure $\Lambda(\Hcal^v)$
is strongly correlated with both the variance and bias
of the corresponding estimator $\hat\Gamma^v$, as shown in the left and center panel of \cref{fig:biasVarianceWeights}.
We also plot the estimators in terms of their bias and variance
in the right panel of that figure.
We observe a clear bias-variance tradeoff,
with the unit vectors $e_m$ yielding low variance but high bias,
and the vector $v_0^*$ yielding low bias but high variance.
Other vectors fall between these two extremes,
depending on their half-space measure $\Lambda(\Hcal^v)$.
This behavior can be explained by implicit censoring, similarly as in the bivariate case above.

\begin{figure}
    \centering
    \resizebox{1\textwidth}{!}{%
        \includegraphics{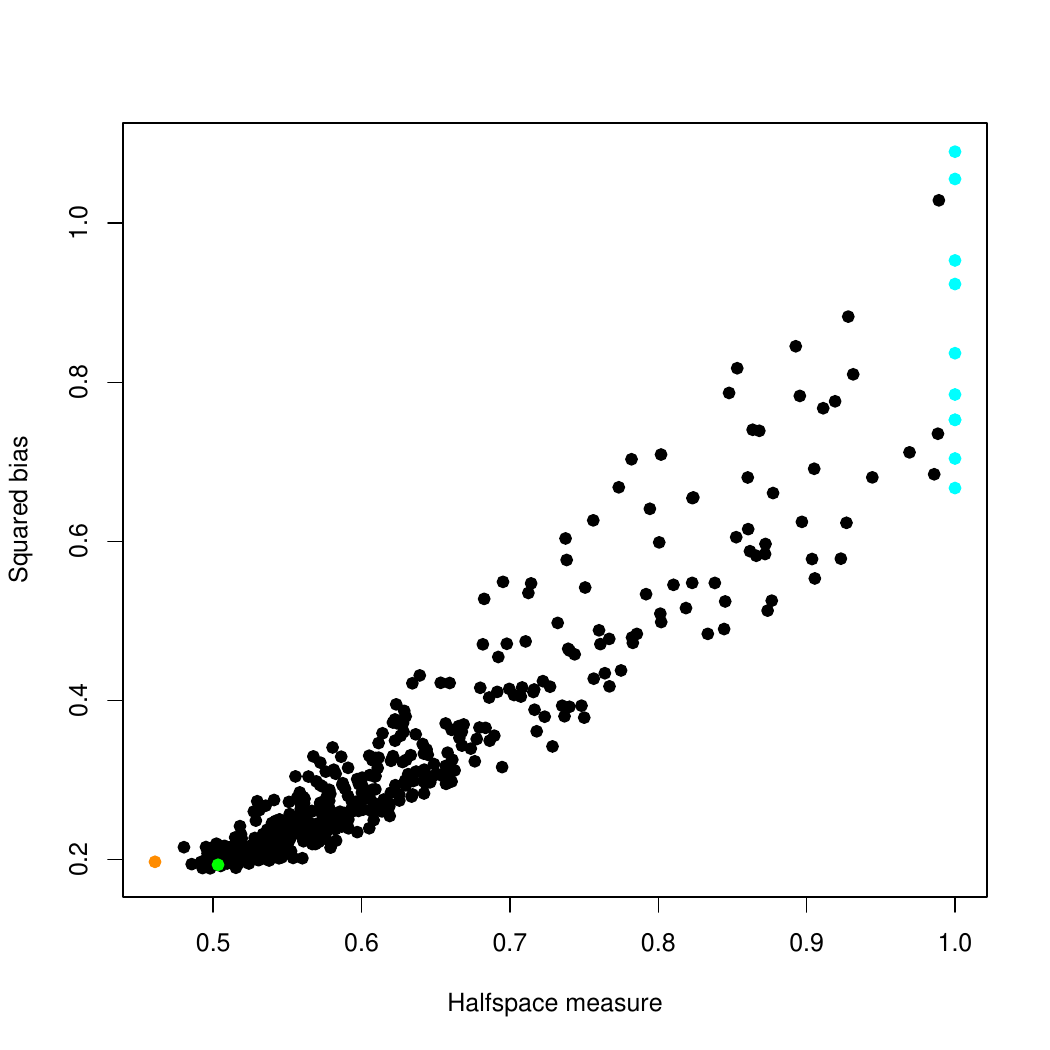}
        \includegraphics{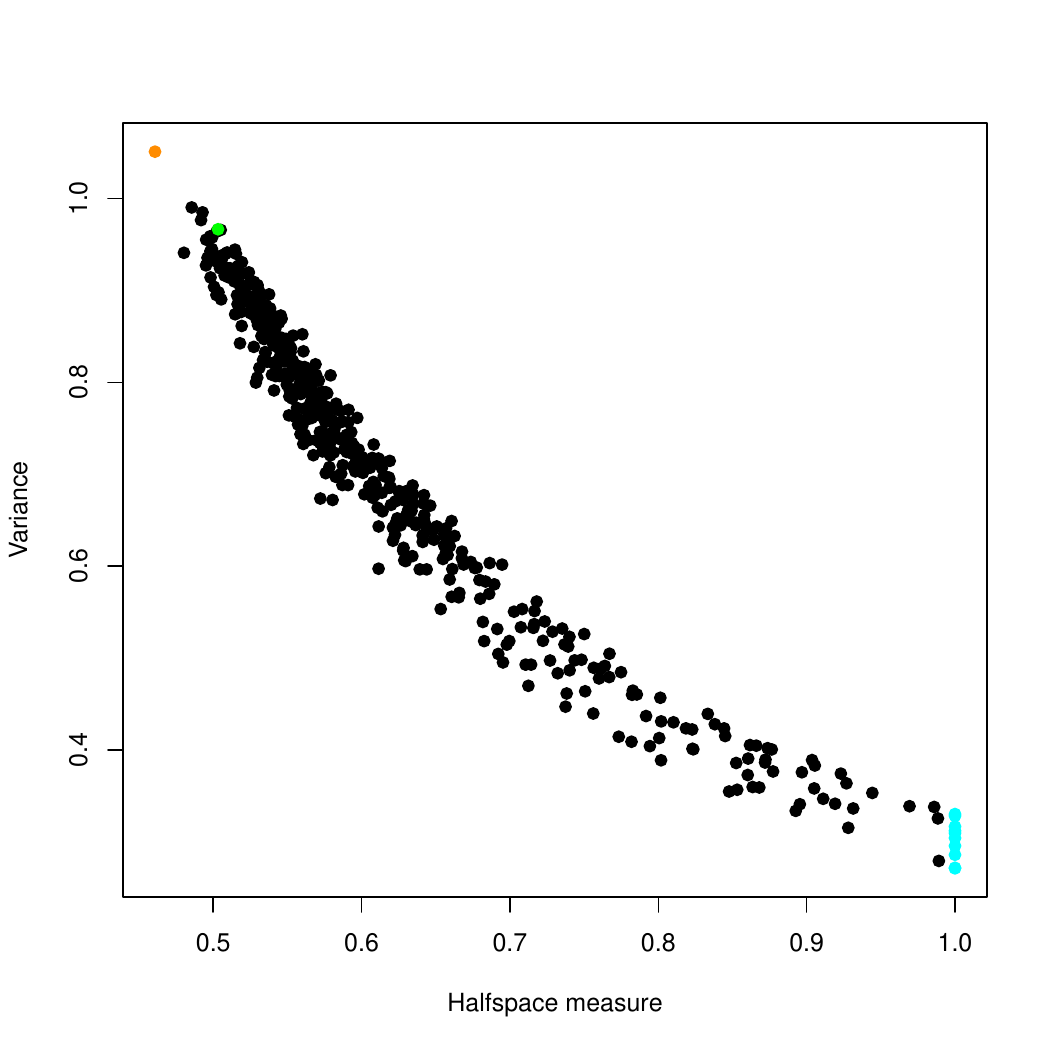}
        \includegraphics{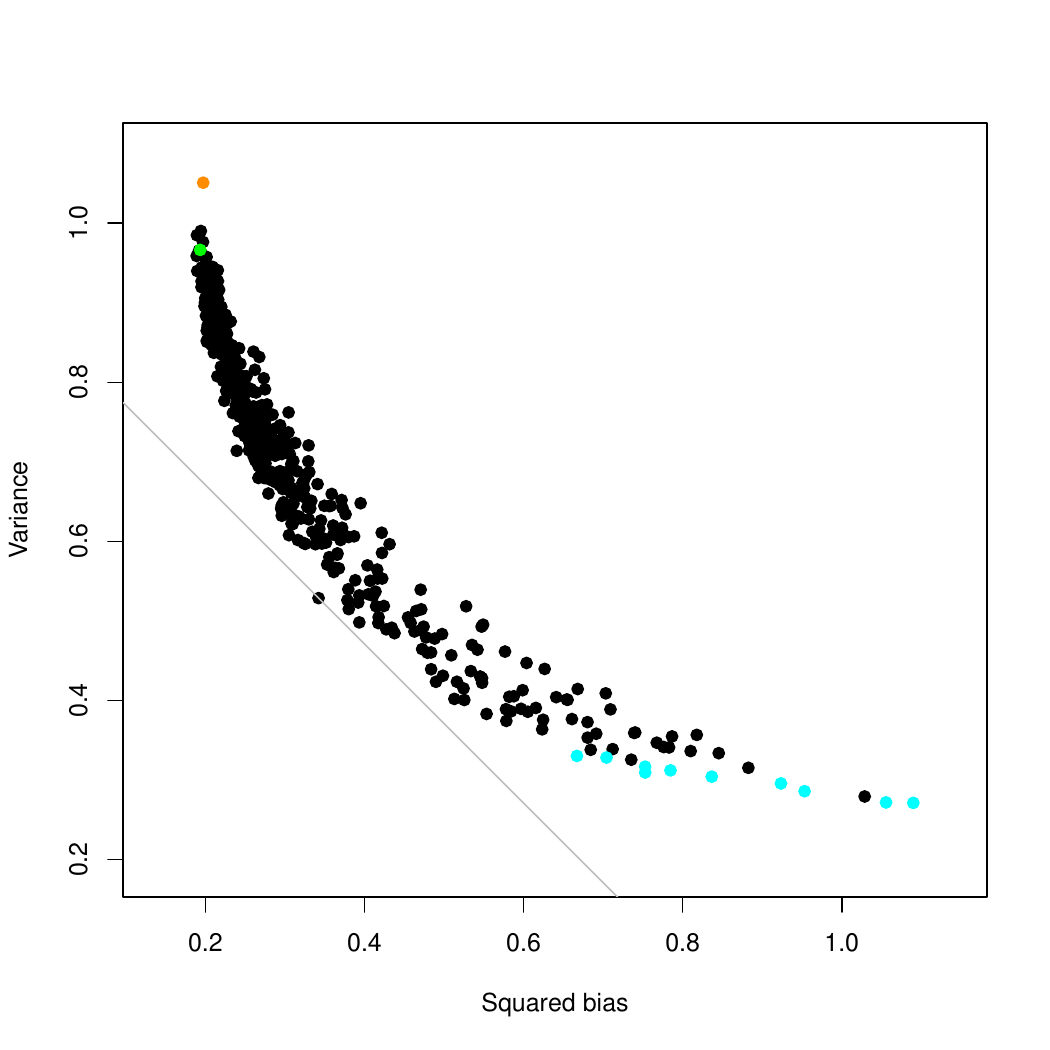}
    }
    \caption{%
        Estimator performance for data in the domain of attraction,
        using $n=2000$ observations per dataset and a threshold of $p=0.98$.
        Each point corresponds to a choice of $v \in \Delta_{d-1}$.
        Shown are the squared bias (left) and variance (center)
        against the half-space measure $\Lambda(\Hcal^v)$,
        and variance against squared bias (right) for the different vectors $v$.
        The unit vectors $e_m$ are shown in cyan,
        $v_0^*$ in orange,
        and $\dinvvec$ in green.
    }
    \label{fig:biasVarianceWeights}
\end{figure}

The optimal choice of vector depends strongly on the setup.
\cref{fig:biasVarianceFixK} shows the bias-variance tradeoff for different threshold
quantiles $p$.
For each dataset, the sample size $n$ is chosen such that a fixed number of $k=100$
observations exceeds the threshold.
This is not a practical scenario,
as $n$ is usually fixed in real applications,
but it allows us to isolate the effect of ``more extreme'' data
from the effect of having less data for higher thresholds.
The level curve of the best MSE is indicated by the gray line,
showing that the quality of the estimation improves as the threshold increases.

\cref{fig:biasVarianceFixN} shows the more realistic scenario of having
a fixed sample size $n=500$ and having fewer threshold exceedances for higher thresholds $p$.
Here, we observe the same general trend,
with respect to the bias-variance tradeoff for different vectors.
However, due to the decreasing number of threshold exceedances for higher thresholds,
the overall performance first improves and then degrades as seen by the increasing level curve of the best MSE.

\begin{figure}
    \centering
    \resizebox{\textwidth}{!}{%
        \includegraphics{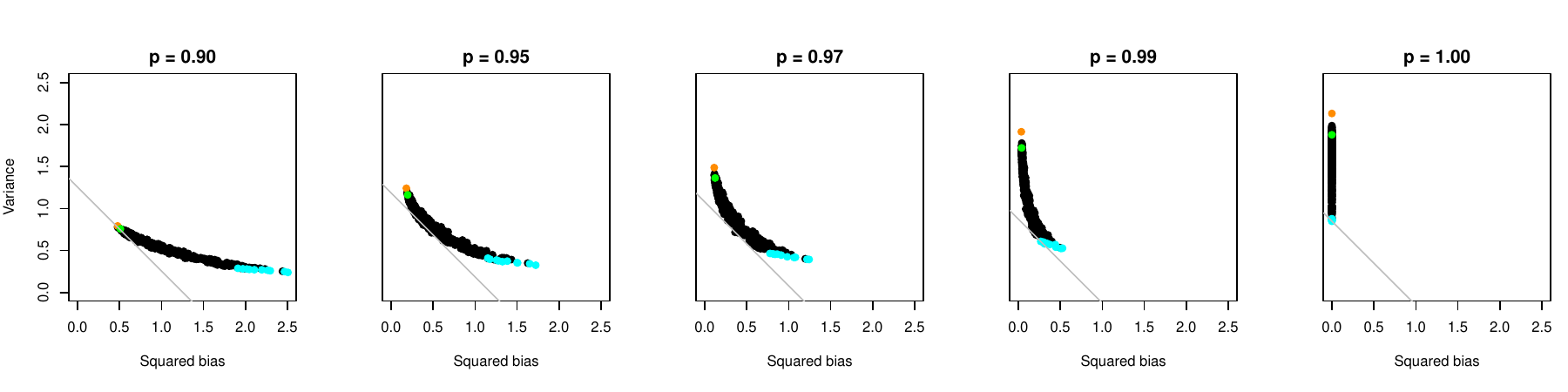}
    }
    \caption{%
        Bias-variance tradeoff for different threshold quantiles $p$,
        with a fixed number of $k=100$ threshold exceedances for each dataset.
        The limiting \multGP{} distribution is included as the $p=1.0$ quantile.
        Highlighted vectors are colored as in \cref{fig:biasVarianceWeights}.
    }
    \label{fig:biasVarianceFixK}
\end{figure}

\begin{figure}
    \centering
    \resizebox{\textwidth}{!}{%
        \includegraphics{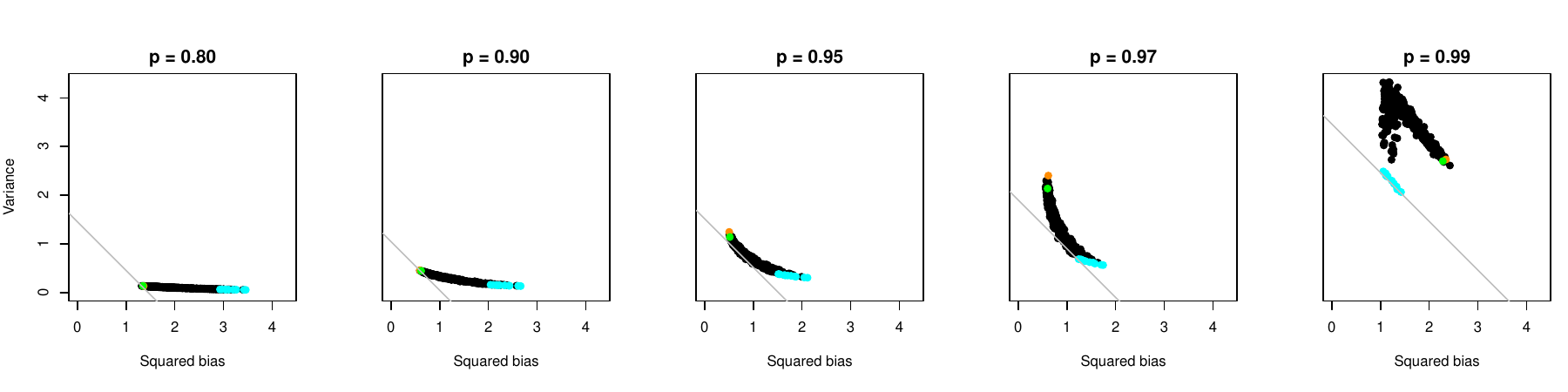}
    }
    \caption{%
        Bias-variance tradeoff for different threshold quantiles $p$,
        with a fixed sample size of $n=500$ for each dataset.
        Highlighted vectors are colored as in \cref{fig:biasVarianceWeights}.
    }
    \label{fig:biasVarianceFixN}
\end{figure}

In both \cref{fig:biasVarianceFixK,fig:biasVarianceFixN},
we observe that for low thresholds $p$,
the bias dominates and vectors with smaller half-space measures perform better.
As the threshold increases,
the variance starts to dominate and for large thresholds,
the unit vectors $e_m$ perform best as they use the most data.
In the limiting model,
included as the $p=1.0$ quantile in \cref{fig:biasVarianceFixK},
the bias is zero and the performance of the estimators only depends on the variance.

\subsection{Ensemble estimators}

Similar to the empirical variogram from \cite{engelkeVolgushev2022},
we define ensemble estimators by averaging over multiple $v$-variograms $\hat\Gamma^v$.
Based on the strong correlation between the half-space measure $\Lambda(\Hcal^v)$
and the bias and variance of the corresponding estimator $\hat\Gamma^v$,
we consider a grouping strategy based on the half-space measures $\Lambda(\Hcal^v)$,
dividing the vectors into ten groups of equal size based on their half-space measure.
As this quantity is not known in practice,
we also include a grouping strategy based on the sparsity of the vectors $v$,
which happens to be a reasonable proxy for $\Lambda(\Hcal^v)$.
The empirical variogram from \cite{engelkeVolgushev2022} is included here
as the ensemble of vectors with sparsity $d-1$,
each containing only a single non-zero entry.
Furthermore, we consider an ensemble of all vectors from above.
In total, we consider 21 different ensemble estimators:
ten based on half-space measure, ten based on sparsity,
and the ensemble of all vectors.

\begin{figure}
    \centering
    \resizebox{\textwidth}{!}{%
        \includegraphics{figures/simulation/HR/ensembles/variances_vs_biases2_with_level_curve.pdf}
        \includegraphics{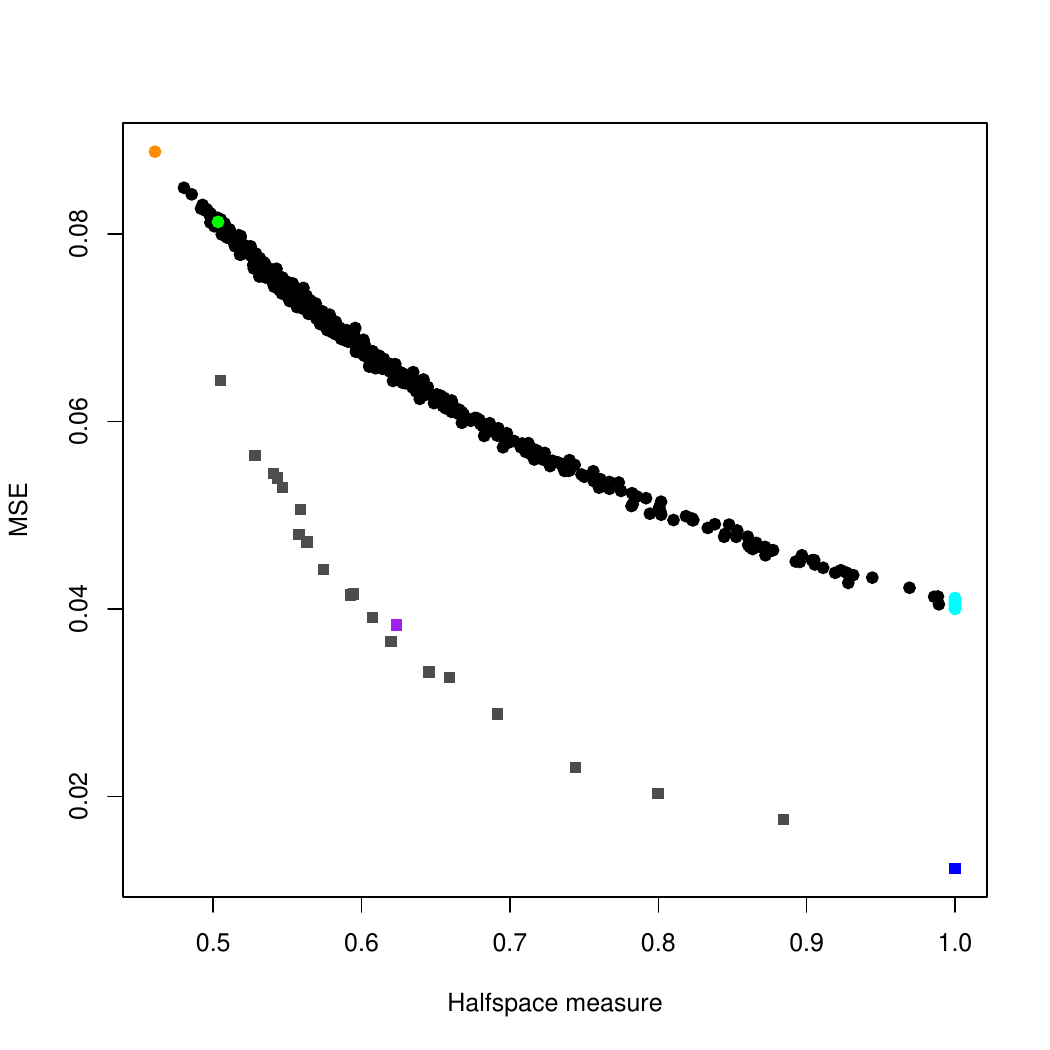}
    }
    \caption{%
        Performance of ensemble estimators,
        compared to the performance of single vector estimators.
        Individual vectors are highlighted as in \cref{fig:biasVarianceWeights}.
        Bags are shown as squares,
        the ensemble of all 500 vectors $v$
        is highlighted in purple,
        and the empirical variogram from \cite{engelkeVolgushev2022} in blue.
        Left: data in the domain of attraction with $n=2000$ and $p=0.98$.
        Right: data from the limiting \multGP{} distribution with $n=2000$.
        Note that we show MSE vs.~half-space measure for the limiting distribution,
        as in this scenario all estimators are unbiased.
    }
    \label{fig:ensembleBiasVariance}
\end{figure}

\cref{fig:ensembleBiasVariance} shows that the ensemble estimators
perform significantly better than single vector estimators,
reducing their variance and averaging the biases.
This is in line with expectations,
since different $\hat\Gamma^v$ are based on different but overlapping sets of observations,
leading to positively correlated estimators.
Since we grouped the vectors based on their half-space measures,
the same bias-variance tradeoff from before is preserved for the ensemble estimators,
with the optimal choice of vectors depending on the setup,
see also \cref{fig:biasVarianceFixKWithBags}.

\begin{figure}
    \centering
    \resizebox{\textwidth}{!}{%
        \includegraphics{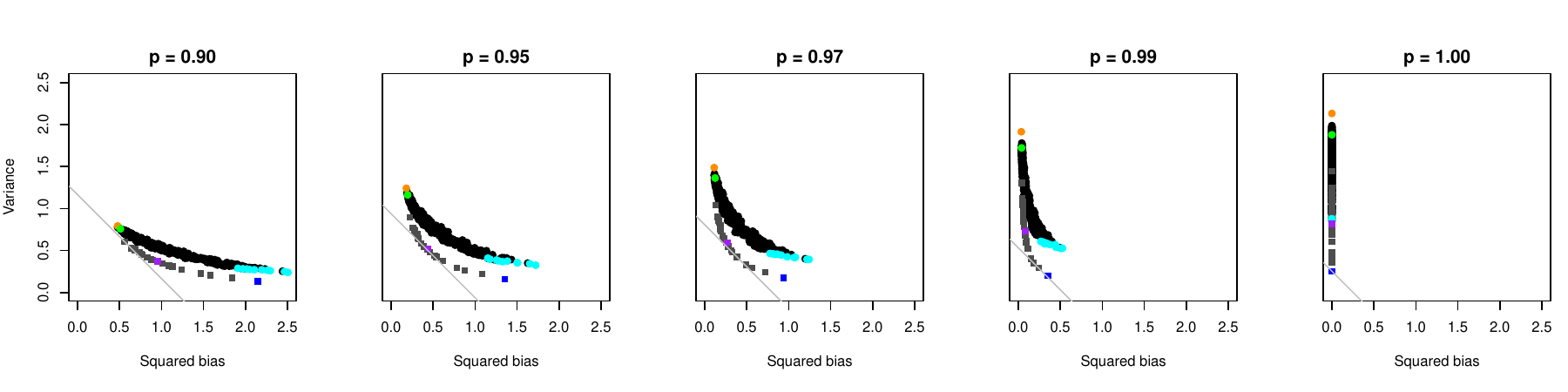}
    }
    \caption{%
        Bias-variance tradeoff for different threshold quantiles $p$
        and fixed number of $k=100$ threshold exceedances for each dataset,
        including ensemble estimators.
        Color highlights and shapes are as in \cref{fig:ensembleBiasVariance}.
    }
    \label{fig:biasVarianceFixKWithBags}
\end{figure}

\section*{Acknowledgements}

The authors thank Ignacio Echave-Sustaeta Rodríguez for helpful discussions and comments.

\clearpage
\appendix

\section{Details on the Poisson point process construction}
\label{sec:ppp_construction}

Consider the Poisson point processes $\Pi^0$ on $\R \times \Rd$ given by
\begin{equation*}
    \Pi^0
    =
    \sum_{i\in\N} \delta_{(\xi_i, U_{(i)})}
\end{equation*}
with intensity $\exp(-x) dx \P[U]$,
where $U$ is a (possibly degenerate)
$d$-dimensional random vector $U$ with $\E(e^{U_i}) = 1$, $i=1,\dots, d$ and $U_{(i)}$ are independent copies of $U$.
From this, we construct the point process $\Pi$ on $\Rd$ as
\begin{equation*}
    \Pi
    =
    \sum_{i\in\N} \delta_{(\xi_i + U_{(i)})}
    ,
\end{equation*}
which has intensity $\Lambda(A) = \int_{\R} e^{-x} \P(U + \onevec x \in A)dx$.
The corresponding max-stable process $Z$ on exponential scale is given by
\begin{equation*}
    Z
    =
    \max_{i\in\N} \xi_i + U_{(i)}
    .
\end{equation*}
The corresponding multivariate Pareto distribution $Y$
is defined by the restriction of $\Pi$ to the set $\Lcal = \setm{y \in \Rd}{y \nleq \zerovec}$,
that is
\begin{equation*}
    \P\lr{Y \in A}
    =
    \Lambda(A \cap \Lcal) / \Lambda(\Lcal)
    .
\end{equation*}
For details on this construction, see for example \citet{resnick2008extreme}.

\section{Details on densities on hyperplanes}
\label{sec:density_details}

When defining densities on hyperplanes,
we consider them with respect to the following dominating measure.
\begin{definition}
    \label{def:hyperplane_measure}
    Let $\lambda_{d-1}$ be the $(d-1)$-dimensional Lebesgue measure on $\Rd[d-1]$.
    For $v\in\Rd$ define $\mu_v$ as
    \begin{align*}
        \mu_v(A) = \lambda_{d-1}(Q_v\T A)
        ,
    \end{align*}
    for measurable $A \subseteq v\orth$,
    and $Q_v \in \Rd[d\times(d-1)]$
    with columns forming an orthonormal basis of $v\orth$.
\end{definition}
Note that the measure is independent of the choice of the orthonormal basis $Q_v$.
This measure is the Hausdorff measure on the hyperplane $v\orth$,
using scaling such that it coincides with
the $(d-1)$-dimensional Lebesgue measure on $\Rd[d-1]$ when $v = \kev{m}$.

\begin{lemma}
    \label{lemma:hyperplane_normal_density}
    Let $v \in \Rd$ and
    $X$ be a multivariate normal random vector
    supported on $v\orth$,
    with mean $\mu \perp v$ and covariance $\Sigma$,
    satisfying $\kernel(\Sigma) = \laspan(v)$.
    Then its density with respect to $\mu_v$ is
    \begin{align*}
        f(x)
        =
        (2\pi)^{-(d-1)/2}
        \pdet{\Sigma}^{-1/2}
        \exp\lr{
            -\half
            (x-\mu)\T
            \Sigma\pinv
            (x-\mu)
        }
        ,
    \end{align*}
    for $x \in v\orth$.
\end{lemma}
\begin{proof}[\proofref{lemma:hyperplane_normal_density}]
    Let the columns of $Q_v \in \Rd[d\times(d-1)]$
    form an orthonormal basis of $v\orth$
    and write $C = Q_v\T \Sigma Q_v \in \Rdd[(d-1)]$.
    Let $Z \sim \normal\lr{\zerovec, C}$
    and observe that
    $X \deq \mu + Q_v Z$, since $Q_v Q_v\T$ is the orthogonal projection matrix onto $v\orth$.
    Since $C$ is invertible,
    $Z$ has density
    \begin{align*}
        f_Z(z)
        =
        (2\pi)^{-(d-1)/2}
        \det{C}^{-1/2}
        \exp\lr{
            -\half
            z\T
            C\inv
            z
        }
        ,
    \end{align*}
    with respect to $\lambda_{d-1}$.
    Hence, the density of $X$ with respect to $\mu_v$ is given by
    \begin{align*}
        f(x)
        =
        f_Z(Q_v\T(x-\mu))
        &=
        (2\pi)^{-(d-1)/2}
        \det{C}^{-1/2}
        \exp\lr{
            -\half
            (x-\mu)\T
            Q_v
            C\inv
            Q_v\T
            (x-\mu)
        }
        \\
        &=
        (2\pi)^{-(d-1)/2}
        \pdet{\Sigma}^{-1/2}
        \exp\lr{
            -\half
            (x-\mu)\T
            \Sigma\pinv
            (x-\mu)
        }
        .
    \end{align*}
    In the last equality we use that $\det{C} = \pdet{\Sigma}$ and
    $\Sigma\pinv = Q_v C\inv Q_v\T$,
    which follows from the fact that $Q_v$ is orthonormal
    and basic properties of the pseudodeterminant and pseudoinverse.
\end{proof}

\begin{lemma}
    \label{lemma:hyperplanes_jacobian}
    Let $u, v \in \Rd$,
    satisfying $\onevec\T u = \onevec\T v = 1$,
    and $P_v = \Id - \onevec v\T$
    be the oblique projection along $\onevec$ onto $v\orth$.
    Let $J_{uv}$ denote the Jacobian determinant of the
    mapping $P_v: u\orth \to v\orth$.
    Then
    \begin{equation*}
        J_{uv} = \frac{\norm{v}}{\norm{u}}
        .
    \end{equation*}
\end{lemma}
\begin{proof}[\proofref{lemma:hyperplanes_jacobian}]
    We give a geometric proof,
    using the fact that the Jacobian determinant of a map
    is the factor by which it scales volumes.
    For a full-rank matrix $A \in \Rd[d\times k]$,
    let $V(A)$ denote the $k$-dimensional volume
    (Hausdorff measure)
    of the parallelepiped
    spanned by its columns.

    Next, consider
    $A_u \in \Rd[d\times(d-1)]$ with columns forming a basis of $u\orth$,
    and the \pppd{} spanned by
    the columns of $A_u$ and $\onevec$.
    Its volume $V(\brackets{A_u, \onevec})$
    can be computed in two ways:
    by multiplying the volume of its base $V(A_u)$
    with its height $\norm{\onevec - (\Id - uu\T/(uu\T)) \onevec}$,
    and
    as the determinant of the matrix $\brackets{A_u, \onevec}$.
    Hence, we have
    \begin{equation*}
        \det{\brackets{A_u, \onevec}}
        =
        V(\brackets{A_u, \onevec})
        =
        V(A_u) \cdot \norm{\onevec - (\Id - uu\T/(uu\T)) \onevec}
        =
        V(A_u) \cdot \norm{u u\T\onevec / (u\T u)}
        =
        V(A_u) / \norm{u}
        ,
    \end{equation*}
    implying $V(A_u) = \norm{u} \cdot \det{\brackets{A_u, \onevec}}$.
    Similarly,
    for $A_v = P_v A_u$ we have
    $V(A_v) = \norm{v} \cdot \det{\brackets{P_v A_u, \onevec}}$.
    Note that the projection $A_u \mapsto P_v A_u$
    adds a multiple of $\onevec$ to each column of $A_u$,
    and hence does not change the determinant of the matrix $\brackets{A_u, \onevec}$,
    i.e.,
    $\det{\brackets{P_v A_u, \onevec}} = \det{\brackets{A_u, \onevec}}$.

    Putting everything together,
    we can compute the Jacobian determinant
    of the mapping $P_v: u\orth \to v\orth$ as
    \begin{equation*}
        J_{uv}
        =
        \frac{V(A_v)}{V(A_u)}
        =
        \frac{
            \norm{v} \cdot \det{\brackets{P_v A_u, \onevec}}
        }{
            \norm{u} \cdot \det{\brackets{A_u, \onevec}}
        }
        =
        \frac{\norm{v}}{\norm{u}}
        .
        \qedhere
    \end{equation*}
\end{proof}

\section{Proofs}
\subsection{Proofs from \cref{sec:stoch_rep}}

\begin{proof}[\proofref{prop:stochrep_halfspace}]
    The proof follows along the lines of the proof of Proposition~3.3 in \cite{wan2025},
    generalizing the vector~$\onevec$ to any $v \in \Delta_{d-1}$,
    and adjusting projections and half-spaces accordingly.
    Let $Q(x) = x - \onevec \max(x)$ denote the (non-linear) projection onto the set
    $\setm{x \in \Rd}{\max(x) = 0}$.
    Let $\tilde W^v \deq P_v Y$,
    supported on~$v\orth$,
    and $S \deq Q(Y) \deq Q(\tilde W^v)$.
    From Theorem~7 in \cite{rootzen2018} we have the stochastic representation
    \begin{align*}
        Y
        &\deq
        S + \onevec E
        \\ &=
        \tilde W^v - \onevec \max(\tilde W^v) + \onevec E
        ,
    \end{align*}
    where $E$ follows a standard exponential distribution,
    and $S$ and $E$ are independent.
    Using $\onevec\T v = 1$, we have
    \begin{equation*}
        v\T Y \geq 0
        \iff
        E - \max(\tilde W^v) \geq 0
    \end{equation*}
    and
    \begin{equation*}
        Y^v
        \deq
        \onevec (E - \max(\tilde W^v)) + \tilde W^v
        \mid
        \set{E \geq \max(\tilde W^v)}
        .
    \end{equation*}
    Next,
    consider sets of the shape $A = B + \onevec \cdot \halfopen{s, \infty}$,
    for $B \subset v\orth$ and $s \in \R$.
    We have
    \begin{align*}
        \P(Y^v \in A)
        &=
        \P(Y \in A \mid v\T Y \geq 0)
        \\ &=
        \P(
            \tilde W^v \in B , v\T Y \geq s
            \mid
            v\T Y \geq 0
        )
        \\ &=
        \P(
            \tilde W^v \in B , E - \max(\tilde W^v) \geq s
            \mid
            v\T Y \geq 0
        )
        \\ &=
        \frac{
            \int_s^\infty
            e^{-t}
            \P(\tilde W^v \in B , \max(\tilde W^v) \leq t - s)
            dt
        }{
            \int_0^\infty
            e^{-t}
            \P(\max(\tilde W^v) \leq t)
            dt
        }
        \\ &=
        \frac{
            \int_0^\infty
            e^{-(u+s)}
            \P(\tilde W^v \in B , \max(\tilde W^v) \leq u)
            du
        }{
            \int_0^\infty
            e^{-t}
            \P(\max(\tilde W^v) \leq t)
            dt
        }
        \\ &=
        e^{-s}
        \cdot
        \frac{
            \int_0^\infty
            e^{-u}
            \P(\tilde W^v \in B , \max(\tilde W^v) \leq u)
            du
        }{
            \int_0^\infty
            e^{-t}
            \P(\max(\tilde W^v) \leq t)
            dt
        }
        .
    \end{align*}
    For $B = v\orth$ we get
    \begin{equation*}
        \P(E - \max(\tilde W^v) \geq s \mid v\T Y \geq 0)
        =
        e^{-s}
        ,
    \end{equation*}
    and for $s=0$ we obtain
    \begin{equation*}
        \P(\tilde W^v \in B \mid v\T Y \geq 0)
        =
        \frac{
            \int_0^\infty
            e^{-u}
            \P(\tilde W^v \in B , \max(\tilde W^v) \leq u)
            du
        }{
            \int_0^\infty
            e^{-t}
            \P(\max(\tilde W^v) \leq t)
            dt
        }
        .
    \end{equation*}
    Hence, the distribution of $Y^v$ decouples into the distribution of
    $W^v = \tilde W^v \mid \set{v\T Y \geq 0}$
    and an independent standard exponential random variable.
\end{proof}

\begin{proof}[\proofref{prop:generator_tilting}]
    The distribution of $Y^v$ can be expressed
    by restricting the exponent measure $\Lambda$ to the half-space $\Hcal^v$.
    By plugging in the expression for $\Lambda$ from \cref{eq:exponent_measure},
    we then get, for any Borel set $A \subset \Rd$,
    \begin{align}
        \nonumber
        \P\lr{Y^v \in A}
        =
        \frac{\Lambda(A \cap \Hcal^v)}{\Lambda(\Hcal^v)}
        &=
        \frac{1}{\Lambda(\Hcal^v)}
        \int_{\R}
        e^{-x}
        \P\lr{U+\onevec x \in A, v\T U + x > 0}
        dx
        \\ &=
        \nonumber
        \frac{1}{\Lambda(\Hcal^v)}
        \int_{\R}
        e^{-x}
        \E\lr{\indicator{U + \onevec x \in A, x > -v\T U}}
        dx
        \\ &=
        \label{eq:Yv_distribution}
        \frac{1}{\Lambda(\Hcal^v)}
        \E\lr{
        \int_{-v\T U}^\infty
        e^{-x}
        \indicator{U + \onevec x \in A}
        dx}.
    \end{align}
    Next,
    we use the exponentially tilted random vector $U^v$
    to construct a random vector $\tilde Y^v$ as $\tilde Y^v := \onevec E + U^v - \onevec v\T U^v$,
    with an independent standard exponential random variable $E$,
    which has distribution
    \begin{align*}
        \P\lr{\tilde Y^v \in A}
        &=
        \int_0^\infty
        e^{-x}
        \P\lr{U^v + \onevec x - \onevec v\T U^v \in A}
        dx
        \\ &=
        \frac{1}{\E\lr{\exp\lr{v\T U}}}
        \int_0^\infty
        e^{-x}
        \E\lr{
            \exp\lr{v\T U}
            \indicator{U + \onevec x - \onevec v\T U \in A}
        }
        dx
        \\ &=
        \frac{1}{\E\lr{\exp\lr{v\T U}}}
        \E\lr{
        \int_0^\infty
        e^{-(x-v\T U)}
        \indicator{U + \onevec (x - v\T U) \in A}
        dx}
        \\ &=
        \frac{1}{\E\lr{\exp\lr{v\T U}}}
        \E\lr{
        \int_{-v\T U}^\infty
        e^{-x}
        \indicator{U + \onevec x \in A}
        dx}
        .
    \end{align*}
    Up to the normalizing constant,
    this is the same as \cref{eq:Yv_distribution}.
    Since both are probability distributions with the same support,
    their normalizing constants must be equal,
    i.e.,
    \begin{equation}
        \label{eq:proportionality_constant}
        \Lambda(\Hcal^v) = \E\lr{\exp\lr{v\T U}},
    \end{equation}
    and we have the stochastic representation
    \begin{equation*}
        Y^v
        \deq
        \onevec E + U^v - \onevec v\T U^v
        ,
    \end{equation*}
    which yields the result after applying the projection $P_v$ to both sides.
    The density expression follows directly from the definition of exponential tilting.
\end{proof}

\begin{proof}[\proofref{cor:extremal_functions_are_generators}]
    We check the exponent measure $\Lambda^{W^v}$ obtained from \cref{eq:exponent_measure}
    when plugging in $W^v + \onevec c$ for some $c \in \R$ as the generator $U$.
    Using the representation $W^v \deq P_v U^v$ from \cref{prop:generator_tilting},
    we obtain, for any Borel set $A \subset \Rd$,
    \begin{align*}
        \Lambda^{W^v}(A)
        &=
        \int_{\R}
        e^{-x}
        \P(W^v + \onevec c + \onevec x \in A)
        dx
        \\ &=
        \int_{\R}
        e^{-x}
        \P(U^v + \onevec (c + x - v\T U^v) \in A)
        dx
        \\ &=
        \int_{\R}
        e^{-x}
        \frac{1}{\E\lr{\exp\lr{v\T U}}}
        \E\lr{
            \exp\lr{v\T U}
            \indicator{U + \onevec (c + x - v\T U) \in A}
        }
        dx
        \\ &=
        \frac{1}{\E\lr{\exp\lr{v\T U}}}
        \E\lr{
            \int_{\R}
            e^{-(x-v\T U)}
            \indicator{U + \onevec (c + x - v\T U) \in A}
            dx
        }
        \\ &=
        \frac{1}{\E\lr{\exp\lr{v\T U}}}
        \E\lr{
            \int_{\R}
            e^{-x + c}
            \indicator{U + \onevec x \in A}
            dx
        }
        \\ &=
        \frac{
            \exp\lr{c}
        }{
            \E\lr{\exp\lr{v\T U}}
        }
        \Lambda(A)
        .
    \end{align*}
    In the third equality we use the definition of $U^v$
    and in the fifth equality we use the change of variable $x \mapsto x + v\T U - c$.
    For $c = \log\lr{\E\lr{\exp\lr{v\T U}}}$ we obtain $\Lambda^{W^v} \equiv \Lambda$.
\end{proof}

\begin{proof}[\proofref{cor:density_Yv_from_lambda}]
    The proportionality of $f_{Y^v}$ and $\lambda$ follows from the fact
    that $Y^v$ is defined as the restriction of $Y$ to the half-space $\Hcal^v$, while the distribution of $Y$ is equal to $\Lambda$ restricted to $\Lcal$ and normalized to a probability measure, where $\Hcal^v \subset \Lcal$.
    The value of the proportionality constant $\Lambda(\Hcal^v) = \E\lr{\exp\lr{v\T U}}$ is given in \cref{eq:proportionality_constant}.
    The second expression follows by rearranging the first one for $\lambda$,
    using different vectors $u$ and $v$.
\end{proof}

\begin{proof}[\proofref{corr:extremal_function_tilting}]
    Consider $U^u$ and $U^v$ as in \cref{prop:generator_tilting}.
    Since consecutive exponential tilting is equivalent to summing the tilting vectors,
    we have that
    $U^v$ is the exponentially tilted version of $U^u$ with tilting vector $v-u$.
    Let $A^0 = A^v + \onevec \R$ and note that $P_u\inv(A^u) = P_v\inv(A^v) = A^0$.
    Then, for the extremal functions $W^u$ and $W^v$,
    we have
    \begin{align*}
        \P(W^v \in A^v)
        &=
        \P(U^v \in A^0)
        \\ &=
        \frac{1}{
            \E\lr{\exp\lr{(v-u)\T U^u}}
        }
        \E\lr{
            \exp\lr{(v-u)\T U^u}
            \indicator{U^u \in A^0}
        }
        \\ &=
        \frac{1}{
            \E\lr{\exp\lr{v\T W^u}}
        }
        \E\lr{
            \exp\lr{v\T W^u}
            \indicator{W^u \in A^u}
        }
        ,
    \end{align*}
    using in the last equality that
    \begin{equation}
        \label{eq:extremal_function_tilting_identity}
        v\T W^u
        =
        v\T P_u U^u
        =
        v\T(U^u - \onevec u\T U^u)
        =
        (v-u)\T U^u
        .
    \end{equation}
    To compute the density of $W^v$,
    consider $\tilde W^u = P_u W^v$,
    which satisfies
    \begin{align*}
        \P(\tilde W^u \in A^u)
        &=
        \P(W^v \in A^v)
        \\ &=
        \frac{1}{
            \E\lr{\exp\lr{v\T W^u}}
        }
        \E\lr{
            \exp\lr{v\T W^u}
            \indicator{W^u \in A^u}
        }
        .
    \end{align*}
    This is an exponentially tilted version of $W^u$ with tilting vector $v$
    and density
    \begin{equation*}
        f_{\tilde W^u}(x)
        =
        \frac{1}{
            \E\lr{\exp\lr{v\T W^u}}
        }
        \exp\lr{v\T x}
        f_{W^u}(x)
        , \qquad
        x \in u\orth
        .
    \end{equation*}
    The distribution of $W^v$ is then obtained by projecting $\tilde W^u$ back to $v\orth$,
    which is a linear transformation with
    inverse $P_u$ and Jacobian $\norm{v}/\norm{u}$
    (see \cref{lemma:hyperplanes_jacobian}).
    Its density on $v^\perp$ is therefore given by
    \begin{equation*}
        f_{W^v}(y)
        =
        \frac{\norm{u}}{\norm{v}} f_{\tilde W^u}(P_u y)
        =
        \frac{\norm{u}}{\norm{v}}
        \frac{1}{
            \E\lr{\exp\lr{v\T W^u}}
        }
        \exp\lr{v\T P_u y}
        f_{W^u}(P_u y)
        ,
    \end{equation*}
    which is the stated density after simplifying $v\T P_u y = v\T y - v\T \onevec u\T y = -u\T y$.
\end{proof}

\begin{proof}[\proofref{prop:density_Yv}]
    The proof goes by showing that the statement is true for all $u \in \Delta_{d-1}$ if it is true for one $v \in \Delta_{d-1}$,
    and then verifying the statement for the specific choice $v = \dinvvec$.
    To apply \cref{cor:density_Yv_from_lambda},
    we consider $y \in \Rd$ satisfying $v\T y > 0$ and $u\T y > 0$.
    Since $f_{Y^u}$ must satisfy homogeneity along the $\onevec$-direction (cf. \cref{sec:MPD})
    and any $y \in \Hcal^u$ can be shifted by a multiple of $\onevec$ to satisfy $v\T y > 1$,
    the density expression also holds for $y \in \Hcal^u$ with $v\T y \leq 0$.

    Assume that the density expression from the result holds for some $v \in \Delta_{d-1}$,
    and apply \cref{cor:density_Yv_from_lambda,corr:extremal_function_tilting}
    to obtain
    \begin{align}
        f_{Y^u}(y)
        &=
        \frac{
            \E\lr{\exp\lr{v\T U}}
        }{
            \E\lr{\exp\lr{u\T U}}
        }
        f_{Y^v}(y)
        \nonumber
        \\ &=
        \frac{
            \E\lr{\exp\lr{v\T U}}
        }{
            \E\lr{\exp\lr{u\T U}}
        }
        f_{W^v}(P_v y)
        \norm{v}
        \exp\lr{-v\T y}
        \nonumber
        \\ &=
        \frac{
            \E\lr{\exp\lr{v\T U}}
        }{
            \E\lr{\exp\lr{u\T U}}
        }
        \frac{
            \norm{u}
        }{
            \norm{v}
        }
        \frac{1}{
            \E\lr{\exp\lr{v\T W^u}}
        }
        \exp\lr{-u\T P_v y}
        f_{W^u}(P_u y)
        \norm{v}
        \exp\lr{-v\T y}
        \nonumber
        \\ &=
        \frac{
            \E\lr{\exp\lr{v\T U}}
        }{
            \E\lr{\exp\lr{u\T U}}
            \E\lr{\exp\lr{v\T W^u}}
        }
        \norm{u}
        \exp\lr{-u\T y + u\T\onevec v\T y}
        \exp\lr{-v\T y}
        f_{W^u}(P_u y)
        \nonumber
        \\ &=
        \frac{
            \E\lr{\exp\lr{v\T U}}
        }{
            \E\lr{\exp\lr{u\T U}}
            \E\lr{\exp\lr{v\T W^u}}
        }
        \norm{u}
        \exp\lr{-u\T y}
        f_{W^u}(P_u y)
        .
        \label{eq:density_Yu_intermediate}
    \end{align}
    To simplify the initial fraction in this expression,
    we use \cref{eq:extremal_function_tilting_identity}
    and the definition of $U^u$ as the exponentially tilted version of $U$ with tilting vector $u$,
    yielding
    \begin{align*}
        \E\lr{\exp\lr{v\T W^u}}
        &=
        \E\lr{\exp\lr{(v-u)\T U^u}}
        \\ &=
        \frac{1}{
            \E\lr{\exp\lr{u\T U}}
        }
        \E\lr{
            \exp\lr{u\T U}
            \exp\lr{(v-u)\T U}
        }
        \\ &=
        \frac{1}{
            \E\lr{\exp\lr{u\T U}}
        }
        \E\lr{
            \exp\lr{v\T U}
        }
        .
    \end{align*}
    The initial fraction in \cref{eq:density_Yu_intermediate}
    therefore simplifies to $1$,
    leaving
    \begin{equation*}
        f_{Y^u}(y)
        =
        \norm{u}
        \exp\lr{-u\T y}
        f_{W^u}(P_u y)
        ,
    \end{equation*}
    which is the stated density.

    Lastly, we verify the ``base case'' $v = \dinvvec$.
    The densities of $W^{\onevec}$ and $E \onevec$ can be multiplied directly,
    since the spaces they are supported on are orthogonal (see \cref{fig:densityProjection}).
    Note that the multiplication of $E$ with the vector $\onevec$
    has inverse $\dinvvec\T (\onevec E)=E$, with Jacobian $d\inv = \norm{\dinvvec}$,
    yielding
    \begin{equation*}
        f_{Y^{\onevec}}(y)
        =
        f_{W^{\onevec}}(P_{\onevec} y)
        \norm{\dinvvec}
        \exp\lr{-\dinvvec\T y}
        . \qedhere
    \end{equation*}
\end{proof}

\begin{proof}[\proofref{cor:v_variogram_expression}]
    We have
    \begin{equation*}
        Y^v
        \deq
        W^v + \onevec E
        \deq
        U^v + \onevec (E - v\T U^v)
    \end{equation*}
    and hence
    \begin{equation*}
        Y_i^v - Y_j^v
        \deq
        W_i^v - W_j^v
        \deq
        U_i^v - U_j^v
        ,
    \end{equation*}
    which implies that the variances of the three differences are equal.
\end{proof}

\begin{proof}[\proofref{ex:variogram_logistic}]
    From the proof of \citet[][Proposition 6]{dombry2016},
    we know that the logistic distribution with parameter $\theta \in (0, 1)$
    has generator $U = (U_1, \dots, U_d)$
    where $U_i$ are independent with Gumbel distribution with scale $\theta$
    and location $-\log(\Gamma(1-\theta))$.

    Following \cref{prop:generator_tilting}, the multivariate exponential tilting by $\exp(v\T U)$,
    applied to their product density,
    is then equivalent to a univariate exponential tilting of each $U_i$ by $\exp(v_i U_i)$,
    i.e.,
    \begin{equation}
        \label{eq:logistic_tilting}
        \exp(v\T U)
        \prod_{i=1}^d
        f_i(x_i)
        =
        \prod_{i=1}^d
        \exp(v_i x_i)
        f_i(x_i)
        .
    \end{equation}
    The variance of an
    exponentially tilted random variable can be computed as the
    second derivative of the cumulant-generating function of the original variable,
    evaluated at the tilting parameter $v_i$.
    For the considered Gumbel random variables with scale $\theta$ this computation yields
    \begin{equation*}
        \Var\lr{U^v_i}
        =
        \theta^2
        \psi^{(1)}\lr{1 - v_i \theta}
        .
    \end{equation*}
    By \cref{cor:v_variogram_expression} and the independence of the components $U_i^v$,
    the variogram of $Y^v$ is then given by
    \begin{align*}
        \Gamma_{ij}^v
        &=
        \Var\lr{U^v_i - U^v_j}
        \\ &=
        \Var\lr{U^v_i} + \Var\lr{U^v_j}
        \\ &=
        \theta^2
        \psi^{(1)}\lr{1 - v_i \theta}
        +
        \theta^2
        \psi^{(1)}\lr{1 - v_j \theta}
        . \qedhere
    \end{align*}
\end{proof}

\begin{proof}[\proofref{ex:variogram_dirichlet}]
    From \citet[][Theorem/Definition 2.3]{Corradini2024}
    we know the following generator $U$ for the Dirichlet model with parameters $\alpha_1, \dots, \alpha_d$:
    \begin{equation*}
        U = \log\lr{
            \tilde U_1, \dots, \tilde U_d
        }
        ,
    \end{equation*}
    where $\tilde U_i$ are independent $\text{Gamma}(\alpha_i, 1/\alpha_i)$ random variables.
    The $i$th entry of $U$ then has univariate (exp-Gamma) density
    \begin{equation*}
        f_i(x)
        =
        \frac{1}{\Gamma(\alpha_i) (1/\alpha_i)^{\alpha_i}}
        \exp\lr{
            \alpha_i x
        }
        \cdot
        \exp\lr{
            -\alpha_i \exp{x}
        }
        .
    \end{equation*}
    Applying the exponential tilting to their product density
    is equivalent to a univariate exponential tilting of each $U_i$ by $v_i$
    (cf. \cref{eq:logistic_tilting}).
    For the density of the tilted random vector $U^v$ we then have
    \begin{equation*}
        \P\lr{U^v \in dx}
        \propto
        \prod_{i=1}^d
        \exp\lr{v_i x_i}
        \cdot
        \exp\lr{
            \alpha_i x_i
        }
        \cdot
        \exp\lr{
            -\alpha_i \exp{x_i}
        }
        ,
    \end{equation*}
    which is again a product of exp-Gamma densities,
    now with parameters $\alpha_i + v_i$ and $1/\alpha_i$.

    Using the fact that the variance of an exponentially tilted random variable
    can be computed as the second derivative of the cumulant-generating function of the original variable,
    evaluated at the tilting parameter $v_i$,
    and an expression for the cumulant-generating function of an exp-Gamma random variable from \citet{Halliwell2021Log},
    we obtain
    \begin{equation*}
        \Var\lr{U^v_i}
        =
        \psi^{(1)}(\alpha_i + v_i)
        .
    \end{equation*}
    Computing the variogram of $Y^v$ using \cref{cor:v_variogram_expression}
    and the independence of the components $U_i^v$
    then yields
    \begin{align*}
        \Gamma_{ij}^v
        &=
        \Var\lr{U^v_i - U^v_j}
        \\ &=
        \Var\lr{U^v_i} + \Var\lr{U^v_j}
        \\ &=
        \psi^{(1)}(\alpha_i + v_i)
        +
        \psi^{(1)}(\alpha_j + v_j)
        . \qedhere
    \end{align*}
\end{proof}

\begin{proof}[\proofref{ex:HR}]
    The generator of the \HR{} distribution with parameter matrix $\Gamma$ can be chosen as any multivariate normal distribution with covariance matrix $\Sigma$ and mean vector $\mu = -\frac12 d_\Sigma$ whose variogram~\eqref{vario_def} is $\Gamma$ \citep[][Section 4.3]{engelke2020}.
    It is well-known that exponential tilting of a normal random vector
    corresponds to shifting the mean and leaving the covariance matrix unchanged.
    Therefore, we have
    \begin{equation*}
        U^v
        \sim
        \normal\lr{\mu + \Sigma v, \Sigma},
    \end{equation*}
    and subsequent multiplication with $P_v$ yields the stated distribution
    \begin{align*}
        W^v
        &\sim
        \normal\lr{P_v \mu + P_v \Sigma v, P_v \Sigma P_v\T}
        \\ &\sim
        \normal\lr{\mu_v, \Sigma^v}
        .
    \end{align*}
    In order to confirm the expression for $\mu_v$ and $\Sigma^v$ in terms of $\Gamma$,
    we use the identities $P_v \onevec = \zerovec$ and $\onevec\T v = 1$ to simplify
    \begin{align*}
        -\half P_v \Gamma P_v\T
        &=
        -\half P_v \lr{d_\Sigma \onevec\T + \onevec d_\Sigma\T - 2 \Sigma} P_v\T
        \\ &=
        -\half P_v (-2 \Sigma) P_v\T
        \\ &=
        P_v \Sigma P_v\T
        , \\
        -\half P_v \Gamma v
        &=
        -\half P_v \lr{d_\Sigma \onevec\T + \onevec d_\Sigma\T - 2 \Sigma} v
        \\ &=
        -\half P_v (d_\Sigma - 2 \Sigma v)
        \\ &=
        P_v \mu + P_v \Sigma v
        .
    \end{align*}
    To confirm the stated variogram expression, observe that
    \begin{equation*}
        \lr{\Sigma^v}_{ij}
        =
        e_i\T P_v \Sigma P_v\T e_j
        =
        \Sigma_{ij} - (\Sigma v)_i - (\Sigma v)_j + v\T \Sigma v,
    \end{equation*}
    and hence
    \begin{equation*}
        \Gamma^v_{ij}
        =
        \lr{\Sigma^v}_{ii} + \lr{\Sigma^v}_{jj} - 2 \lr{\Sigma^v}_{ij}
        =
        \Sigma_{ii} + \Sigma_{jj} - 2 \Sigma_{ij}
        =
        \Gamma_{ij}
        ,
    \end{equation*}
    since all expressions involving $v$ cancel out.
\end{proof}

\subsection{Proofs from \cref{sec:density}}

\begin{proof}[\proofref{lemma:integralHyper}]
    It is to be shown that $\E(\exp(v\T U)) = \exp\lr{-\frac{1}{4} v\T \Gamma v}$,
    for a valid generator~$U$ of the \HR{} distribution with parameter matrix $\Gamma$.
    From \cref{ex:HR}, we know that a normally distributed random vector $U$ with covariance matrix $\Sigma$ and mean vector $-\half d_\Sigma$
    is such a generator
    provided it has variogram matrix $\Gamma$, i.e., it satisfies
    \begin{equation*}
        \Gamma
        =
        d_\Sigma \onevec\T + \onevec d_\Sigma\T - 2\Sigma
        ,
    \end{equation*}
    where $d_\Sigma$ is the vector of diagonal entries of $\Sigma$.
    For such a random vector $U$, we have
    \begin{align*}
        \exp\lr{
            -\tfrac{1}{4} v\T \Gamma v
        }
        &=
        \exp\lr{
            -\tfrac{1}{4}
            v\T\lr{
                d_\Sigma \onevec\T + \onevec d_\Sigma\T - 2\Sigma
            }v
        }
        \\ &=
        \exp\lr{
            -\half v\T d_\Sigma
            +
            \half v\T \Sigma v
        }
        \\ &=
        \E\lr{\exp(v\T U)}
        ,
    \end{align*}
    using $\onevec\T v = 1$ and the moment generating function of the multivariate normal distribution.
\end{proof}

\begin{proof}[\proofref{lem:pseudoDeterminants}]
    With $\Sigma = P_{\onevec}(-\half\Gamma)P_{\onevec}$,
    let $X \in \Rd[d \times (d-1)]$ be such that $\Sigma = X X\T$.
    Let $v \in \Rd$ be such that $\onevec\T v = 1$,
    and denote $\delta := v - \dinvvec \perp \onevec$.
    Then $\Sigma^{v} = P_v \Sigma P_v\T = P_v X X\T P_v\T$,
    and
    \begin{align}
        \notag
        \pdet{\Sigma^{v}}
        &=
        \pdet{P_v X X\T P_v\T}
        \\ &=
        \pdet{P_v\T P_v X X\T}
        .
        \label{eq:SPPXX}
    \end{align}
    For the first two factors in \cref{eq:SPPXX}, we have
    \begin{align*}
        M^v
        &:=
        P_v\T P_v
        \\ &=
        (\Id - v \onevec\T)(\Id - \onevec v\T)
        \\ &=
        \Id - \onevec v\T - v \onevec\T + v \onevec\T \onevec v\T
        \\ &=
        \Id - \dinvvec\onevec\T + d \delta \delta\T
        ,
    \end{align*}
    which has the following eigenvectors:
    \begin{align*}
        M^v\onevec
        &=
        \onevec - \onevec + \zerovec
        = \zerovec
        , \\
        M^v\delta
        &=
        \delta - \zerovec + d \delta \delta\T \delta
        =
        (1 + d\delta\T \delta) \delta
        , \\
        M^v w
        &=
        w + \zerovec + \zerovec
        = w
        , \quad \text{for }
        w \perp \onevec, \delta
        .
    \end{align*}
    Hence, the eigenvalues of $M^v$ are $0$, $1 + d\delta\T \delta$, and $1$ with multiplicity $d-2$,
    yielding pseudo-determinant
    \begin{equation*}
        \pdet{M^v}
        =
        1 + d\delta\T \delta
        =
        dv\T v
        =
        d \norm{v}^2
        .
    \end{equation*}
    The last two factors in \cref{eq:SPPXX} give $X X\T = \Sigma$.
    Since both $M^v$ and $\Sigma$ are symmetric with identical kernel $\laspan\set{\onevec}$,
    we have $\pdet{M^v \Sigma} = \pdet{M^v} \pdet{\Sigma}$.
    This can be seen by considering a basis change with an orthogonal basis containing $\onevec$,
    which transforms both matrices into block diagonal matrices with a zero block of size $1$,
    and invertible blocks of size $d-1$, to which the multiplicativity of the regular determinant applies.

    Putting together the above considerations for $v,w$, with $\onevec\T w = \onevec\T v = 1$,
    we get
    \begin{equation*}
        \pdet{\Sigma^{v}}
        =
        d\norm{v}^2 \cdot \pdet{\Sigma}
    \end{equation*}
    and hence
    \begin{equation*}
        \frac{
            \pdet{\Sigma^{v}}
        }{
            \pdet{\Sigma^{w}}
        }
        =
        \frac{\norm{v}^2}{\norm{w}^2}
        . \qedhere
    \end{equation*}
\end{proof}

\begin{proof}[\proofref{prop:HrDensity}]
    First, we establish that $B := P_v\T (\Sigma^v)\pinv P_v$
    is a pseudoinverse of $\Sigma$
    by verifying the conditions of Lemma~S.1.4 in \cite{hentschel2023}.
    Note that
    \begin{align*}
        \onevec\T B
        &=
        \onevec\T (\Id - v\onevec\T) (\Sigma^v)\pinv (\Id - \onevec v\T)
        =
        \zerovec
        \\ \Rightarrow \quad
        \image(B)
        &\subseteq
        \onevec\orth
        =
        \image(\Sigma)
         .
    \end{align*}
    Using
    $\Sigma = P_{\onevec} \Sigma^v P_{\onevec}$
    and a series of simplifications of projection matrices,
    we have
    \begin{align*}
        \Sigma B
        &=
        P_{\onevec} \Sigma^v P_{\onevec} P_v\T (\Sigma^v)\pinv P_v
        \\ &=
        P_{\onevec} \Sigma^v( \Sigma^v)\pinv P_v
        \\ &=
        P_{\onevec} (\Id - vv\T/\norm{v}^2) P_v
        \\ &=
        P_{\onevec}
        .
    \end{align*}
    Hence,
    the conditions of the Lemma are satisfied,
    and we have $B = \Sigma\pinv = \Theta$,
    yielding
    \begin{equation*}
        \norm[(\Sigma^v)\pinv]{P_v x}^2
        =
        (P_v x)\T( \Sigma^v)\pinv (P_v x)
        =
        x\T B x
        =
        \norm[\Theta]{x}^2
        .
    \end{equation*}
    Furthermore, \cref{lem:pseudoDeterminants} yields $\pdet{\Sigma^v} = d\norm{v}^2\pdet{\Sigma}$,
    and a basic property of pseudo-determinants and pseudoinverses is
    $\pdet{\Sigma}\inv = \pdet{\Sigma\pinv} = \pdet{\Theta}$,
    implying
    \begin{equation*}
        \sqrt{\pdet{\Sigma^v}\inv}\cdot\norm{v}
        =
        \sqrt{d\inv \norm{v}^{-2} \pdet{\Sigma}\inv}\cdot\norm{v}
        =
        \sqrt{d\inv \pdet{\Theta}}
        .
    \end{equation*}
    Combining these two identities with the expression for the density of $Y^v$ in \cref{eq:density_HR_Yv},
    we get
    \begin{align*}
        f_{Y^v}(y)
        &=
        \sqrt{(2\pi)^{-(d-1)}\pdet{\Sigma^v}\inv}
        \exp\lr{
            -\half\norm[(\Sigma^v)\pinv]{P_v y-\mu_v}^2
        }
        \norm{v}
        \exp\lr{-v\T y}
        \\ &=
        \sqrt{d\inv(2\pi)^{-(d-1)}\pdet{\Theta}}
        \exp\lr{
            -\half
            \norm[\Theta]{y-\mu_v}
        }
        \exp\lr{-v\T y}
        ,
    \end{align*}
    with $\mu_v$ as defined in \cref{ex:HR}.
    Note that $\Theta\mu_v = \Theta P_{\onevec}\mu_v = \Theta\tilde\mu_v$,
    which is why $\mu_v$ can be replaced by $P_\onevec \mu_v = \tilde\mu_v$ in the expression above.
    To express $\mu_v$ and hence $\tilde\mu_v$ independently of the choice of $\Sigma$,
    observe that
    \begin{equation*}
        P_v\lr{-\half\Gamma}v
        =
        P_v\lr{
            -\half d_\Sigma \onevec\T - \half\onevec d_\Sigma\T + \Sigma
        }v
        =
        P_v\lr{
            -\half d_\Sigma + \Sigma v
        }
        =
        \mu_v
        .
    \end{equation*}

    To obtain the expression for the exponent measure density $\lambda$,
    we use \cref{cor:density_Yv_from_lambda,lemma:integralHyper} to get
    for all $y \in \Hcal^v$ that
    \begin{align*}
        \lambda(y)
        &=
        f_{Y^v}(y) \,
        \E\lr{\exp\lr{v\T U}}
        \\ &=
        \sqrt{d\inv(2\pi)^{-(d-1)}\pdet{\Theta}}
        \exp\lr{-\tfrac{1}{4}v\T\Gamma v}
        \exp\lr{
            -\half\norm[\Theta]{y-\mu_v}^2
        }
        \exp\lr{-v\T y}
        .
    \end{align*}
    Since $\lambda$ must be homogeneous,
    this expression also holds for all $y \in \Rd$.
\end{proof}

\begin{proof}[\proofref{prop:HRdensityLimit}]
    As $v_0\T\onevec=1$, we can plug in $v_0$ in the density expression in \cref{prop:HrDensity}.
    This yields
    \begin{equation*}
        \lambda(y; \Theta)
        =
        \sqrt{d\inv\lr{2\pi}^{-(d-1)}\pdet{\Theta}}
        \exp\lr{-\tfrac{1}{4}v_0\T\Gamma v_0}
        \cdot
        \exp\lr{
            - v_0\T y
        }
        \cdot
        \exp\lr{
            - \half
            \norm[\Theta]{
                y - \mu_{v_0}
            }^2
        }
        .
    \end{equation*}
    We have
    $\frac{1}{4}v_0\T\Gamma v_0=t_0^2\onevec\T\Gamma^{-1}\Gamma\Gamma^{-1}\onevec=\half t_0$
    and
    $\mu_{v_0}= P_{\onevec} (-\half \Gamma) v_0=-t_0P\Gamma\Gamma^{-1}\onevec=\zerovec$,
    which yields
    \begin{equation*}
        \lambda(y; \Theta)
        =
        \sqrt{d\inv\lr{2\pi}^{-(d-1)}\pdet{\Theta}}
        \exp\lr{-\half t_0}
        \cdot
        \exp\lr{
            - v_0\T y
        }
        \cdot
        \exp\lr{
            - \half
            \norm[\Theta]{
                y
            }^2
        }
        .
    \end{equation*}
    It follows that for $v_0\ge \zerovec$ we have $Y^{v_0}=W^{v_0}+E\onevec$ with $W^{v_0}\sim\normal(\zerovec,\Sigma^{v_0})$, where
    \begin{align*}
        \Sigma^{v_0}&=P_{v_0}\Sigma P_{v_0}\T
        \\ &=
        P_{v_0}P_{\onevec}(-\half\Gamma)P_{\onevec}\T P_{v_0}\T
        \\ &=
        P_{v_0}(-\half\Gamma) P_{v_0}\T
        \\ &=
        -\half\Gamma
        +\half \onevec v_0\T\Gamma
        +\half \Gamma v_0\onevec\T
        -\half \onevec v_0\T\Gamma v_0\onevec\T
        \\ &=
        -\half\Gamma + t_0\onevec\onevec\T + t_0\onevec\onevec\T - t_0\onevec\onevec\T
        \\ &=
        -\half\Gamma + t_0\onevec\onevec\T
        . \qedhere
    \end{align*}
\end{proof}

\begin{proof}[\proofref{cor:minIntegral}]
    Let $t_0$ and $v_0$ be the resistance radius and curvature with respect to a conditionally negative definite variogram matrix $\Gamma$.
    It follows from \citet[Proposition~6.14]{devriendt2022a} that the minimum of $\Lambda\lr{\Hcal^v}=-\frac{1}{2}v\T\Gamma v$ over all $v\in\Rd$ with $v\T\onevec=1$ is $-t_0$,
    uniquely achieved at $v = v_0$. Furthermore, because $\Gamma$ has zero diagonal and positive entries otherwise, the quadratic form $v\T\Gamma v$ is nonnegative for $v\in\Delta_{d-1}$ and only vanishes when $v$ is a canonical unit vector. This gives the second inequality.
\end{proof}

\bibliographystyle{apalikeurl}
\bibliography{literature}

\end{document}